\pdfoutput=1
\RequirePackage{ifpdf}
\ifpdf % We~are running pdfTeX in pdf mode
\documentclass[pdftex]{sigma}%,draft
\else
\documentclass{sigma}
\fi

\numberwithin{equation}{section}

\newtheorem{Theorem}{Theorem}[section]
\newtheorem{Corollary}[Theorem]{Corollary}
\newtheorem{Lemma}[Theorem]{Lemma}
\newtheorem{Proposition}[Theorem]{Proposition}
{ \theoremstyle{definition}
\newtheorem{Definition}[Theorem]{Definition}

\newtheorem{Example}[Theorem]{Example}
 }

\usepackage{physics}
\usepackage{tikz-cd}

\newcommand{\set}[2]{\big\{#1\mid #2\big\}}
\newcommand{\from}{\colon}
\newcommand{\id}{\mathrm{id}}
\newcommand{\im}{\mathrm{im}}
\newcommand{\pr}{\mathrm{pr}}
\newcommand{\Z}{\mathbb{Z}}
\newcommand{\R}{\mathbb{R}}
\newcommand{\C}{\mathbb{C}}
\newcommand{\Cun}{\C^{\underline{n}}}
\newcommand{\struc}{\C^\times\wr~I_n}
\newcommand{\conn}{(\C^\times)^n}
\newcommand{\End}{\operatorname{End}}
\newcommand{\Aut}{\operatorname{Aut}}
\newcommand{\Hol}{\operatorname{Hol}}
\newcommand{\GL}{\mathrm{GL}}
\newcommand{\Ad}{\operatorname{Ad}}
\newcommand{\g}{\mathfrak{g}}
\newcommand{\Fr}{\operatorname{Fr}}
\newcommand{\PFr}{\operatorname{PFr}}
\newcommand{\Loop}{\operatorname{Loop}}
\newcommand{\dif}{\mathrm{d}}
\newcommand{\pder}[2]{\frac{\partial #1}{\partial #2}}
\providecommand{\dtzero}{\frac{\dif}{\dif t}\bigg|_{t=0}}
\newcommand{\diag}{\operatorname{diag}}
\newcommand{\discrim}{\operatorname{discrim}}
\providecommand{\cov}{\operatorname{cov}}
\providecommand{\Spec}{\operatorname{Spec}}
\DeclareMathOperator{\Eig}{Eig}
\newcommand{\En}{\mathrm{E}}
\newcommand{\EigFr}{\mathrm{EigFr}}
\newcommand{\G}{\mathcal{G}}
\newcommand{\M}{\mathcal{M}}
\newcommand{\inn}[2]{\langle#1|#2\rangle}
\newcommand{\PT}{\mathcal{PT}}

\begin{document}
\allowdisplaybreaks

\newcommand{\arXivNumber}{2107.02497}

\renewcommand{\PaperNumber}{003}

\FirstPageHeading

\ShortArticleName{A Unified View on Geometric Phases and Exceptional Points}

\ArticleName{A Unified View on Geometric Phases and Exceptional\\ Points in Adiabatic Quantum Mechanics}

\Author{Eric J.~PAP~$^{\rm ab}$, Dani\"el BOER~$^{\rm b}$ and Holger WAALKENS~$^{\rm a}$}
\AuthorNameForHeading{E.J.~Pap, D.~Boer and H.~Waalkens}

\Address{$^{\rm a)}$~Bernoulli Institute, University of Groningen,\\
\hphantom{$^{\rm a)}$}~P.O.~Box 407, 9700 AK Groningen, The Netherlands}
\EmailD{\href{mailto:e.j.pap@rug.nl}{e.j.pap@rug.nl}, \href{mailto:h.waalkens@rug.nl}{h.waalkens@rug.nl}}
\URLaddressD{\url{http://www.rug.nl/staff/e.j.pap/}, \url{http://www.rug.nl/staff/h.waalkens/}}

\Address{$^{\rm b)}$~Van Swinderen Institute, University of Groningen, 9747 AG Groningen, The Netherlands}
\EmailD{\href{mailto:d.boer@rug.nl}{d.boer@rug.nl}}
\URLaddressD{\url{http://www.rug.nl/staff/d.boer/}}

\ArticleDates{Received July 23, 2021, in final form December 28, 2021; Published online January 13, 2022}

\Abstract{We present a formal geometric framework for the study of adiabatic quantum mechanics for arbitrary finite-dimensional non-degenerate Hamiltonians. This framework generalizes earlier holonomy interpretations of the geometric phase to non-cyclic states appearing for non-Hermitian Hamiltonians. We start with an investigation of the space of non-degenerate operators on a finite-dimensional state space. We then show how the energy bands of a Hamiltonian family form a covering space. Likewise, we show that the eigenrays form a bundle, a generalization of a principal bundle, which admits a natural connection yielding the (generalized) geometric phase. This bundle provides in addition a natural generalization of the quantum geometric tensor and derived tensors, and we show how it can incorporate the non-geometric dynamical phase as well. We finish by demonstrating how the bundle can be recast as a principal bundle, so that both the geometric phases and the permutations of eigenstates can be expressed simultaneously by means of standard holonomy theory.}


\Classification{81Q70; 81Q12; 55R99}

\section{Introduction} \label{sec:introduction}

The eigenvalues of a matrix or operator are of great interest in many fields of mathematics and physics. In quantum mechanics, the eigenvalues of an observable are the possible measurement outcomes for this observable. The dynamics is governed by the Hamilton operator, whose eigenvalues define the energy levels. This relation between quantum mechanics and linear algebra, or more general functional analysis, is even tighter in the subfield of adiabatic quantum mechanics. In general, if the Hamiltonian changes in time, then initial eigenstates need not evolve into instantaneous eigenstates of the Hamiltonian at a later time. However, this could be realized for sufficiently slow change of the Hamiltonian \cite{Born1928BeweisAdiabatensatzes,Kato1950OnMechanics}, known as the \emph{adiabatic approximation}.

The interest in adiabatic quantum mechanics gained momentum with the discovery of the geometric phase by Berry \cite{Berry1984QuantalChanges}. He showed that the phase picked up by an eigenstate upon varying parameters along a closed loop has in addition to the usual dynamical contribution a purely geometric contribution. This geometric phase depends only on the traversed loop, i.e., it is invariant under reparametrization of the path. Geometric phases have been found in numerous physical systems, see the overviews in \cite{Chruscinski2004GeometricMechanics,Cohen2019GeometricAndbeyond,Wilczek1989GeometricPhysics}, where in fact the earliest observation of such a phase had been reported by Pancharatnam \cite{Pancharatnam1956GeneralizedApplications} in polarized light. On the theoretical side, the geometric phase was readily identified to be the holonomy of a bundle of eigenstates over system parameters \cite{Simon1983HolonomyPhase}. In addition, generalizations have been studied, e.g., for degenerate Hamiltonians \cite{Wilczek1984AppearanceSystems} or for any so called cyclic state returning to its original ray in a Hilbert space~\cite{Aharonov1987PhaseEvolution}.

Geometric phase is also studied for non-Hermitian Hamiltonians, which allow for gain and loss of energy. Using left-eigenstates instead of bra's, a generalization of Berry's phase for cyclic eigenstates of a non-degenerate non-Hermitian Hamiltonian was presented in \cite{Garrison1988ComplexSystems}, and this was related to a geometric model in \cite{Mehri-Dehnavi2008GeometricInterpretation}. Another prominent feature of a non-Hermitian Hamiltonian family is that, when one follows a loop in the space of system parameters, the energies may swap places, which renders the evolution non-cyclic. This effect indicates a non-trivial topology of the energy bands; they need not be separated as in the Hermitian case. The interchange of energies can be related to degeneracies in the space of system parameters, which are known as exceptional points (EPs). These were first mentioned, with slightly different meaning, in \cite{Kato1966PerturbationOperators}, and have been found in many experimental setups, see the overviews, e.g., in~\cite{Heiss2012ThePoints,Miri2019ExceptionalPhotonics}. Many of these experiments revolve around parity-time or $\PT$ symmetric systems. Such systems are typically described by non-Hermitian Hamiltonians, and the breaking of this symmetry is often associated with an EP. Although $\PT$-symmetry is not necessary for EP theory, it does provide experimentally accessible realizations.

We remark that this phenomenology is not adiabatic in the standard way, as non-adiabatic effects cannot be avoided. The adiabatic theorem in the non-Hermitian case only holds for the state with highest relative gain \cite{Nenciu1992OnHamiltonians}. This leads to significant limitations on how well the state exchange around an EP can be measured \cite{Berry2011SlowPhenomenon,Uzdin2011OnPoints}. A solution is to consider quasi-static set-ups; the path in parameter space is discretized into points, and one measures per configuration, see also \cite{Holler2020Non-HermitianPoints,Milburn2015GeneralPoints}.

In this paper we introduce a geometric formalism to properly describe the adiabatic evolution of eigenstates of any non-degenerate finite-dimensional operator playing the role of a~Hamiltonian. In particular, the eigenstates need not be cyclic and the Hamiltonian need not be Hermitian.
To this end we introduce a bundle that directly follows from the eigenvalue problem and consists of triples ``matrix-eigenvalue-eigenvector''. This bundle will not be principal, rather it has the structure of a semi-principal bundle as rigorously defined and studied in \cite{Pap2020FramesTheory}. It is naturally equipped with a connection, whose parallel transport corresponds to the generalized geometric phase.
As a result, both the geometric phase \emph{and} the swaps of eigenstates arising from EPs can be described in a single holonomy description. In fact, the same bundle also allows for the incorporation of the non-geometric dynamical phase. We also treat its associated ``frame'' bundle, which is principal and provides a rigorous geometric argument behind the matrices used to describing state evolution.

The condition of non-degeneracy of the operator will be crucial, and hence we study the space this defines in Section~\ref{sec:non-deg space}. The geometry of the eigenvalues of these operators we study in Section~\ref{sec:eigval bundle and EPs}. This will yield a natural way to treat the interchanges of energies around EPs. In Section~\ref{sec:eigvec bundle and GPs}, this is extended to include eigenvectors, which yields a natural parallel transport theory for geometric phases, also in the presence of EPs. Following up on this, we show in Section~\ref{sec:holonomy} how these physical phenomena can be interpreted via a single holonomy description. We finish with a discussion in Section~\ref{sec:discussion}.

\section{The non-degeneracy space} \label{sec:non-deg space}

We start with the mathematical objects that will provide the basis for all our arguments concerning eigenvalues. At this moment, we do not consider any specific operator family, instead only a complex vector space $V$ of finite dimension $n$ is given. An important remark is that we do \emph{not} endow $V$ with an inner product; we only use the topological vector space and manifold properties of $V$.

Let us establish some notation by reviewing the following definitions and facts. We recall that endomorphisms or operators, i.e., linear maps $A\from V \to V$, form the space $\End(V)$, which is a complex manifold of complex dimension $n^2$. The set of all eigenvalues of an operator $A$ is called the spectrum of $A$, and we denote it as $\Spec(A)$. If $(v_1,\dots,v_n)$ is a frame, or basis, of~$V$, we call it an eigenframe of~$A$ if each $v_i$ is an eigenvector of $A$. The set of all bases of $V$ we denote as $\Fr(V)$, and the set of all eigenframes of $A$ by $\EigFr(A)$. Any basis $(v_1,\dots,v_n)$ of $V$ has a dual basis $\big(\theta^1,\dots,\theta^n\big)$ of the dual space $V^\vee$ defined by the condition
\begin{equation*}
 \theta^i(v_j)=\delta^i_j=
 \begin{cases}1 & \text{if}\quad i=j,\\ 0 & \text{if} \quad i\ne j.\end{cases}
\end{equation*}
An eigencovector of $A$ is a non-zero covector $\theta$ such that $\theta A=\lambda \theta$, where $\lambda\in \C$ is necessarily an eigenvalue of $A$. An eigencoframe of $A$ is a frame of $V^\vee$ consisting of eigencovectors of $A$. One can verify the following facts about eigenframes.
\begin{Lemma} \label{lem:eigenframe}
 Let $A\in \End(V)$, then
\begin{enumerate}\itemsep=0pt
 \item[$(1)$] $A$ has an eigenframe if and only if $A$ is diagonalizable,
 \item[$(2)$] $(v_1,\dots,v_n)$ is an eigenframe of $A$ if and only if the dual basis $\big(\theta^1,\dots,\theta^n\big)$ is an eigencoframe.
 \end{enumerate}
\end{Lemma}

We will now focus on the operators $A$ which are non-degenerate, i.e., which have $n$ distinct eigenvalues. The subspace in $\End(V)$ of non-degenerate operators we denote by $N(V)$. We will go over different ways to formally define this subspace. First, we inspect an algebraic argument, which allows for a straightforward result on the manifold properties of $N(V)$. Afterwards, we consider other formulations of non-degeneracy, which will naturally guide us to symmetries and bundle properties of $N(V)$.

\subsection{Discriminant definition}

The first characterisation of $N(V)$ we will inspect is based on the discriminant, and will form our algebraic definition of $N(V)$. Naturally, the eigenvalues of $A\in \End(V)$ are the zeros of the characteristic polynomial $p(A,z)$ of $A$, and whether or not the zeros of a polynomial are distinct can be inferred from the discriminant. That is, one first has the map
\begin{align*}
 p \from\ \End(V) &\to \C[z], \\
 A &\mapsto p(A,z):=\det(z I-A)
 \end{align*}
and by evaluating the discriminant one obtains a composite function
\begin{align*}
 d \from\ \End(V) & \to \C,\\
 A & \mapsto \discrim(p(A,z),z)
 \end{align*}
known as the discriminant of operators. Its zero set
\begin{equation*}
 \Delta(V)=\set{A \in \End(V)}{d(A)=0}
\end{equation*}
is called the discriminant set and consists of all matrices that are degenerate. Clearly $N(V)$ is the complement of the discriminant set in $\End(V)$, which yields our formal definition of $N(V)$.

\begin{Definition}[non-degeneracy space]
 Given a finite-dimensional complex vector space $V$, its space of non-degenerate operators, or non-degeneracy space, is
 \begin{equation*}
 N(V)=\set{A\in \End(V)}{d(A)\ne 0}.
 \end{equation*}
\end{Definition}

It is well-known that $d$ becomes a polynomial function in the matrix elements of the operator. Hence one may readily conclude the following.
\begin{Lemma} \label{lem:props N(V)}
 The space $N(V)$ is an $($algebraic$)$ open and dense subset of $\End(V)$. In particular, $N(V)$ is a submanifold of complex dimension $n^2$ and real codimension $2$.
\end{Lemma}

We next consider alternative characterizations of the non-degeneracy space.

\subsection{Parametrizing the non-degeneracy space} \label{sec:param of non-deg}

Another way to describe $N(V)$ is by explicitly parametrizing its elements. One such parametrization readily follows from the fact that any $A \in N(V)$ is similar to a diagonal matrix where the diagonal entries are distinct. In~other words, there is a frame $\tilde{f}=(f_1,\dots,f_n)$ of $V$ so that the matrix of $A$ w.r.t.\ $\tilde{f}$ takes the form $\diag(\lambda_1,\dots,\lambda_n)$, where all the $\lambda_i$ are distinct. We will denote the space of tuples $\tilde{\lambda}=(\lambda_1,\dots,\lambda_n)$ of $n$ distinct complex numbers as $\Cun$; more background on this space can be found in Appendix~\ref{sec:distinct numbers}. We abbreviate $\diag(\lambda_1,\dots,\lambda_n)$ as $\diag\big(\tilde{\lambda}\big)$.

The above decomposition can be recast as the following parametrization. For convenience, let us identify a frame $\tilde{f} \in \Fr(V)$ with the map $S_{\tilde{f}} \from \C^n \to V$ given by $(z_1,\dots,z_n)\mapsto \sum_{i=1}^n z_if_i$, which is a linear isomorphism by definition of a frame. The parametrization of $N(V)$ is then formally described using the map
\begin{align*} %\label{eq:vector + value to N(V)}
 \Xi \from\ \Fr(V)\times \Cun &\to N(V),\\
 \big(\tilde{f},\tilde{\lambda}\big) &\mapsto S_{\tilde{f}}\diag\big(\tilde{\lambda}\big) S_{\tilde{f}}^{-1},
 \end{align*}
which is a smooth surjection. Clearly, this map is not injective for two reasons; firstly, $\Xi$ is indifferent concerning a non-zero scaling of the eigenvectors, and secondly, if we permute both the vectors of $\tilde{f}$ and the values in $\tilde{\lambda}$ we obtain the same operator. This is a symmetry of $\Xi$, which we phrase using group actions.

First, writing the group of non-zero complex numbers as $\C^\times$, the scaling symmetry is given by the $(\C^\times)^n$-action
\begin{equation} \label{eq:conn action Fr(V)timesCun}
 \tilde{z}\cdot\big(\tilde{f},\tilde{\lambda}\big) =\big(\tilde{z}\tilde{f},\tilde{\lambda}\big),
\end{equation}
where $\tilde{z}\tilde{f}$ is the entry-wise product $(z_1f_1,\dots,z_nf_n)\in \Fr(V)$. To describe the permutations, let us agree that a permutation acts on a tuple by permuting its entries, e.g., for $\sigma \in S_n$ and a tuple $\tilde{\lambda}\in \Cun$, $\sigma \cdot \tilde{\lambda}=\big(\lambda_{\sigma^{-1}(1)},\dots,\lambda_{\sigma^{-1}(n)}\big)$. Then the permutation action on $\Fr(V)\times \Cun$ is simply
\begin{equation} \label{eq:perm action Fr(V)timesCun}
 \sigma \cdot \big(\tilde{f},\tilde{\lambda}\big) = \big(\sigma \tilde{f},\sigma \tilde{\lambda}\big).
\end{equation}
The scaling and permutation actions merge into a single action of the wreath product $\C^\times\wr I_n$, where $I_n:=\{1,\dots,n\}$ is an index set with $n$ elements. This group is the semi-direct product $\conn \rtimes S_n$, whose defining action of $S_n$ on $\conn$ is exactly the tuple permutation as we used above. The group multiplication of $\C^\times\wr I_n$ reads
\begin{equation*} %\label{eq:mult wreath product}
 (\tilde{z}_1,\sigma_1)\cdot(\tilde{z}_2,\sigma_2)=\big(\tilde{z}_1 (\sigma_1\tilde{z}_2),\sigma_1\sigma_2\big).
\end{equation*}
{\samepage
The action of the wreath product on $\Fr(V)\times \Cun$ is obtained by performing first the permutation and then the scaling, i.e.,
\begin{align}
 \struc \times (\Fr(V)\times \Cun) &\to \Fr(V)\times \Cun,\nonumber
 \\
 (\tilde{z},\sigma)\cdot \big(\tilde{f},\tilde{\lambda}\big) &= \big(\tilde{z}\big(\sigma\tilde{f}\big),\sigma\tilde{\lambda}\big).
 \label{eq:struc action Fr(V) times Cun}
 \end{align}

 }

Let us make some remarks at this point. First, $\struc$ has a faithful representation by complex generalized permutation matrices. These are matrices with a single non-zero complex number in every row and column, or, equivalently, products $DP$ with $D$ a diagonal matrix without zeros on the diagonal and $P$ a standard permutation matrix. Given the conjugating nature of $\Xi$, the appearance of $\struc$ is natural, as the generalized permutation matrices form the stabilizer of the diagonal matrices. Let us also refer to some earlier reports on the wreath product in combination with this parametrization. For example, in \cite{Wojcik2020HomotopyHamiltonians} there is firstly a~quotient by scaling and secondly the quotient by permutations. Although not reported as such, we recognize the wreath product there. In \cite[Lemma~1.1]{Mandel1987ConstructionMatrices}, it is shown that a restriction of $\Xi$ defines a principal bundle for a wreath product group, where the number field need not be $\C$. As we will argue now, the parametrization $\Xi$ also defines a principal $\C^\times\wr I_n$-bundle. Let us use that a projection is a principal bundle if it is given by the quotient of a free and proper action. The latter is readily verified, hence it remains to be shown that $\Xi$ coincides with this projection.
\begin{Lemma} \label{lem:Xi is principal}
 There is a principal bundle
 \begin{equation*}
 \begin{tikzcd}
 \C^\times \wr I_n \ar{r} & \Fr(V)\times \Cun \ar{r}{\Xi}& N(V).
 \end{tikzcd}
 \end{equation*}
\end{Lemma}

\begin{proof}
 As stated, it suffices to show that $\Xi$ and the action quotient map are isomorphic as bundle maps. As we already saw that the action preserves the fibers of $\Xi$, it only remains to show that each group orbit exhausts the fiber in which it lies. Therefore, assume $\Xi\big(\tilde{f},\tilde{\lambda}\big)=\Xi\big(\tilde{f}',\tilde{\lambda}'\big)=:A$. Then both $\tilde{\lambda}$ and $\tilde{\lambda}'$ constitute $\Spec(A)$, hence they are related by a permutation $\sigma \in S_n$ as $\tilde{\lambda}=\sigma\tilde{\lambda}'$. Then $f_i$ needs to be parallel to $f_{\sigma^{-1}(i)}'$ for each $i$, so $\tilde{f}$ and $\sigma\tilde{f}'$ are equal up to an element of $(\C^\times)^n$. However, then $\big(\tilde{f},\tilde{\lambda}\big)$ and $\big(\tilde{f}',\tilde{\lambda}'\big)$ differ only up to an element of~$\C^\times \wr I_n$.
\end{proof}

This result has various consequences. First, it shows that $N(V)$ can be realized as a quotient space. Another important observation concerns the existence of local sections of $\Xi$. These sections provide local moving eigenframes and corresponding smooth eigenvalues, as we will use in Section~\ref{sec:eigvec bundle and GPs}. We state the details in the following corollaries.
\begin{Corollary}
 Identifying $\Fr(V)$ with $\GL(n,\C)$, the space of non-degenerate operators is rea\-li\-zed as a quotient space as
 \begin{equation*}
 N(V) \cong \frac{\GL(n,\C)\times \Cun}{\C^\times\wr I_n}.
 \end{equation*}
\end{Corollary}

\begin{Corollary} \label{cor:local eigvals eigvecs}
 Let $A_0\in N(V)$. There is a neighborhood $U$ of $A_0$ on which one has smooth local eigenvalue functions $\lambda_1,\dots,\lambda_n \from U
 \to \C$, exhausting the spectrum at each point, and smooth local eigenvector functions $f_1,\dots,f_n \from U \to V$. That is, for all $A\in U$ and $i\in I_n$ one has
 \begin{equation*}
 Af_i(A)=\lambda_i(A) f_i(A)
 \end{equation*}
 such that $\Spec(A)=\{\lambda_1(A),\dots,\lambda_n(A)\}$ and $(f_1(A),\dots,f_n(A))$ is a basis of $V$.
 Moreover, the tuple $(\lambda_1(A),\dots,\lambda_n(A))$ may be taken to be any given ordering of $\Spec(A)$, and the tuple $(f_1(A),\dots,f_n(A))$ any eigenframe of $A$ following the same ordering.
\end{Corollary}

\begin{proof}
 The functions $f_i(A)$ and $\lambda_i(A)$ are the components of a local section $s\from U \to \Fr(V)\times\Cun$, which can be taken through an arbitrary point above $A$.
\end{proof}

\subsection{Spectrum map on non-degeneracy space}

Non-degenerate operators can also be characterized based on their spectrum. Indeed, an operator $A\in \End(V)$ is non-degenerate if and only if $\Spec(A)$ consists of $n$ distinct elements. In~other words, $\Spec(A)$ should not be any subset of $\C$, but belong to the set $\binom{\C}{n}$ of all subsets of $\C$ consisting of $n$ distinct elements. Taking the spectrum can thus be written as a map
\begin{align*}
 \Spec \from N(V) &\to \binom{\C}{n},\\
 A &\mapsto \Spec(A).
 \end{align*}
We will now continue by showing that this map $\Spec$ defines a fiber bundle. In this way, we find~$N(V)$ realized as a total space instead of a base space. In addition, the map $\Spec$ will reappear when we discuss EPs; formally it is this map that associates the change in spectrum to a change in Hamiltonian.

To start, let us verify that $\Spec$ is smooth. For the manifold structure on $\binom{\C}{n}$, we follow the idea that the space $\binom{\C}{n}$ can be obtained from $\Cun$ by reducing an (ordered) tuple to the (unordered) set of its elements. We write $q\from \Cun \to \binom{\C}{n}$ for the quotient map, and use the manifold structure on $\binom{\C}{n}$ for which $q$ defines a principal $S_n$-bundle with the standard permutation action of $S_n$ on $\Cun$. For more details, we refer to Appendix~\ref{sec:distinct numbers}.

We can see that $\Spec$ is smooth using a geometric argument, based on the projection $q$ and the map $\Xi$ from the previous part. The key observation is the equality $\Spec\big(\Xi\big(\tilde{f},\tilde{\lambda}\big)\big)=q\big(\tilde{\lambda}\big)$ for any $\big(\tilde{f},\tilde{\lambda}\big)\in \Fr(V)\times \Cun$; by construction $\tilde{\lambda}$ lists the eigenvalues of the operator $\Xi\big(\tilde{f},\tilde{\lambda}\big)$. In~other words, we have the following commutative diagram:
\begin{equation} \label{eq:diagram SpecXi}
 \begin{tikzcd}
 \Fr(V)\times \Cun \ar{r}{\pr_{\Cun}} \ar{d}{\Xi} & \Cun \ar{d}{q}\\
 N(V) \ar{r}{\Spec} & \binom{\C}{n}.
 \end{tikzcd}
\end{equation}
Smoothness of $\Spec$ is now clear; as the upper route in the diagram is smooth and $\Xi$ is a~surjective submersion the claim follows.

Let us continue by discussing the model fiber of $\Spec$. This is facilitated by viewing operators according to their spectral decomposition. That is, any $A\in N(V)$ can be written as a sum $\sum_{i=1}^n\lambda_iP_i$, where $\lambda_i\in \Spec(A)$ and each $P_i\from V\to V$ is a projection. Geometrically, this expresses that $A\in N(V)$ is completely specified by the pairs of eigenvalue and the corresponding eigenrays. Hence, if we fix the spectrum, then only the choice of eigenrays remains. This means that any fiber of $\Spec$ is diffeomorphic to this space of possible choices of eigenrays, or, equivalently, to the space of suitable tuples of projectors. The latter space we can describe in more detail. Clearly, each $P_i$ projects to a one-dimensional subspace in $V$, and together these projectors satisfy $P_iP_j=\delta_{ij}P_i$ and $\sum_{i=1}^n P_i=\id_V$. That is, the tuple $(P_1,\dots,P_n)$ must lie in the space
\begin{equation*} %\label{eq:set of res of id}
 \PFr(V):=\set{(P_1,\dots,P_n)\in \End(V)}{P_iP_j=\delta_{ij}P_i,\, \dim_\C(\im(P_i))=1},
\end{equation*}
which is thus also the model fiber for $\Spec$. One can think of $\PFr(V)$ as the space of all resolutions of the identity which are compatible with non-degenerate operators, i.e., those corresponding to~$n$ rays in $V$. Clearly, $\PFr(V)$ relates to $\Fr(V)$ by sending a basis $(f_1,\dots,f_n)$, with dual basis $\big(\theta^1,\dots,\theta^n\big)$, to the tuple $\big(f_1\theta^1,\dots,f_n\theta^n\big)\in \PFr(V)$. This map is surjective, but not injective as individual scaling of the basis vectors will yield the same projectors. Hence we find that $\Fr(V)/\conn$ is canonically isomorphic to $\PFr(V)$, and it is the former form that we will see in the upcoming proofs.

At this point we have discussed the relevant spaces, but did not yet discuss their symmetry. This symmetry can be found from the observation that operators with the same spectrum differ by a similarity transformation. Let us implement these transformations in the language of the action of $\GL(V)$ on $N(V)$ given by
\begin{align}
 \GL(V)\times N(V) &\to N(V),\nonumber
 \\
 S\cdot A &= SAS^{-1}.
 \label{eq:GL(V)-action N(V)}
\end{align}
The map $\Spec$ is invariant w.r.t.\ this action by the familiar rule $\Spec\big(SAS^{-1}\big)=\Spec(A)$. Hence each fiber of $\Spec$ is a $\GL(V)$-manifold as well, and so we wish to view $\PFr(V)$ as a $\GL(V)$-manifold, also equipped with the conjugation action. Observe that this action on $\PFr(V)$ is naturally inherited from $\Fr(V)$, on which it reads $S\cdot (f_1,\dots,f_n)=(Sf_1,\dots,Sf_n)$. We thus arrive at the following result.

\begin{Lemma} \label{lem:GL(V) action & fiber Spec}
 The $\GL(V)$-action on $N(V)$ is transitive on the fibers of $\Spec$. Moreover, any fiber of $\Spec$ is isomorphic to $\PFr(V)$ as $\GL(V)$-manifolds.
\end{Lemma}
\begin{proof}
 As any two non-degenerate operators with the same spectrum differ by a similarity transformation, the action is transitive on the fibers of $\Spec$. Hence, every fiber of $\Spec$ is a~homogeneous $\GL(V)$-space. The stabilizer subgroup at $A\in N(V)$ consists of the maps that preserve all eigenrays of $A$ individually, hence is isomorphic to $\conn$. In order to parametrize this fiber, let $\tilde{\lambda}$ be an ordering of $\Spec(A)$, and consider $\Xi$ restricted to the subset $\Fr(V)\times\{\tilde{\lambda}\}$. Clearly, this surjects on the fiber of $\Spec$ containing $A$. Moreover, $\Xi$ is equivariant w.r.t.\ the canonical $\GL(V)$-action on $\Fr(V)\times \Cun$. Hence, the fiber of $\Spec$ containing $A$ is isomorphic, as $\GL(V)$-manifold, to the quotient of $\Fr(V)$ by the stabilizer. The stabilizer $\conn$ is now straightforward; it appears via the $\conn$-action given in equation~\eqref{eq:conn action Fr(V)timesCun}. Hence we found a~$\GL(V)$-equivariant isomorphism to $\Fr(V)/\conn$, and so to $\PFr(V)$.
\end{proof}

The model fiber of $\Spec$ is thus $\PFr(V)$, which we view as a $\GL(V)$-manifold, emphasizing its close relation to similarity transformations. We thus wish to prove that $\Spec$ defines a fiber bundle that respects the $\GL(V)$-action. That is, $\Spec$ defines a $\GL(V)$-manifold bundle in the language of \cite{Pap2020FramesTheory}; both fiber and total space are endowed with a $\GL(V)$-action, and local trivializations can be taken $\GL(V)$-equivariant. We then arrive at the following statement, which summarizes the results of this section.

\begin{Proposition} %\label{prop:Spec as bundle map}
 The spectrum map induces the $\GL(V)$-manifold bundle
 \begin{equation*}
 \begin{tikzcd}
 \PFr(V) \ar{r} & N(V) \ar{r}{\Spec} & \binom{\C}{n}.
 \end{tikzcd}
 \end{equation*}
 \end{Proposition}

\begin{proof}
 Pick a point $\{\lambda_1,\dots,\lambda_n\}\in \binom{\C}{n}$, and let $U\subset \binom{\C}{n}$ be a neighborhood of $\{\lambda_1,\dots,\lambda_n\}$ on which a local section $s\from U \to \Cun$ of $q$ is defined. Consider the map
 \begin{align*}
 U\times \Fr(V) &\to N(V),\\
 \big(u,\tilde{f}\big) &\mapsto \Xi\big(\tilde{f},s(u)\big),
 \end{align*}
 which for each $u\in U$ surjects on the fiber of $\Spec$ above $u$. By applying Lemma~\ref{lem:GL(V) action & fiber Spec} fiber-wise we find that the reduced map $U\times \Fr(V)/\conn \to N(V)$ is well-defined. In fact, if we look at the spectral decomposition, this map pairs the tuple $(P_1,\dots,P_n)\in \PFr(V)$ determined by $\tilde{f}$ with the values in the tuple $s(u)$. We thus obtained a $\GL(V)$-equivariant local trivialization of $\Spec$ around $\{\lambda_1,\dots,\lambda_n\}$, hence the claim follows.
\end{proof}

\subsection{Summarizing diagram}

If one takes the bundles defined by $\Xi$ and $\Spec$ plus the diagram in \eqref{eq:diagram SpecXi}, then one readily obtains the following diagram of bundle sequences:
\begin{equation} \label{eq:summarizing diagram}
 \begin{tikzcd}
 (\C^\times)^n \ar{r}\ar{d} & \C^\times \wr I_n \ar{r}\ar{d} & S_n \ar{d}\\
 \Fr(V) \ar{d}\ar{r} & \Fr(V)\times \Cun \ar{r}{\pr_{\Cun}} \ar{d}{\Xi} & \Cun \ar{d}{q}\\
 \PFr(V) \ar{r} & N(V) \ar{r}{\Spec} & \binom{\C}{n}.
 \end{tikzcd}
\end{equation}
The direct product $\Fr(V)\times \Cun$ we view as a bundle over $\Cun$. The remaining bundles are straightforward; on top is the defining decomposition of the wreath product $\C^\times \wr I_n$, on the right the quotient map $q$, and on the left the quotient realization of $\PFr(V)$. Observe that all rows and columns are group-space bundles, i.e., each one is related to a group action. With the exception of $\Spec$ all bundles are principal; $\Spec$ itself defines a bundle of homogeneous $\GL(V)$-spaces.

\section{Eigenvalue bundle and exceptional points} \label{sec:eigval bundle and EPs}

Let us study the geometry that describes how eigenvalues and eigenvectors depend on the operator. We will find that eigenvalues and eigenvectors form bundles over the non-degenerate operators. This space $N(V)$ of non-degenerate operators will reappear as the region free of singularities. In this way, $N(V)$ shows us where we can use results from geometry, in particular concerning parallel transport, which in turn provides a framework for adiabatic quantum mechanics. We will deal with the eigenvalues in this section, and provide an extended similar argument in the next section concerning the eigenvectors.

\subsection{The spectrum bundle} %\label{sec:spectrum bundle}

We start by describing a natural abstract model for the energy bands. Namely, when studying EPs, we want to follow eigenvalues as a function of the operator. We wish to view this in a~geometric way.
For example, we wish to view an eigenvalue function $\lambda=\lambda(A)$, with $\lambda(A)$ an~eigenvalue of the operator $A$, as a local section. Note that such functions $\lambda(A)$ are necessarily local; otherwise EPs could not exist. We will call such functions \emph{local eigenvalues}.

We quickly come to the conclusion that we should restrict the operators. Namely, degenerate operators pose a problem as they will form singularities. Indeed, around a degenerate energy, the energy bands do not resemble a smooth manifold. On the other hand, for $A\in N(V)$ such issues do not occur. As the following lemma shows, the implicit function theorem yields that simple eigenvalues always admit an extension to a local eigenvalue. An immediate consequence is that the restriction from $\End(V)$ to $N(V)$ is minimal in order to obtain a smooth structure.
\begin{Lemma} \label{lem:non-deg <-> non-zero derivative}
 Given $A \in \End(V)$, then $A \in N(V)$ if and only if
 \begin{equation*}
 \pder{p}{z}(A,\lambda_i) \ne0, \qquad \forall \lambda_i \in \Spec(A).
 \end{equation*}
\end{Lemma}

We can now formalize the bundle which has the local eigenvalues as its local sections. That is, its local sections are of the form $A\mapsto (A,\lambda(A))$ with $\lambda(A)$ a local eigenvalue. This could be used to define the bundle bottom-up, but we prefer to use the following more explicit top-down method. Clearly, the total space of the bundle consists of all pairs $(A,\lambda)$ such that $A\in N(V)$ and $\lambda\in \Spec(A)$. We observe that this set is the zero set of the characteristic polynomial map~$p$, restricted to non-degenerate operators. This viewpoint will form our primary definition\footnote{We remark the similarity with the space $\mathcal{M}$ in \cite{Mehri-Dehnavi2008GeometricInterpretation}. However, as we explicitly list the eigenvalue as a~coor\-dinate, it is not a multi-valued function here.} of the bundle, because of its algebraic convenience. We use the name spectrum bundle: the fiber above $A\in N(V)$ is simply $\Spec(A)$, and the term bundle we will justify in Theorem~\ref{thm:Spec(V) bundle over N(V)}.

\begin{Definition}[spectrum bundle]
 Given the vector space $V$, define its \emph{spectrum bundle} to be the space
 \begin{equation*}
 \Spec(V)=\set{(A,\lambda)\in N(V) \times \C}{p(A,\lambda)=0}.
 \end{equation*}
 Furthermore, we write $\pi_\lambda \from \Spec(V)\to N(V)$ for the projection $(A,\lambda) \mapsto A$.
\end{Definition}

Our first step in proving the bundle property of $\pi_\lambda$ is showing that $\Spec(V)$ is a smooth manifold. This readily follows from the derivative characterization of $N(V)$ in Lemma~\ref{lem:non-deg <-> non-zero derivative}.

\begin{Proposition} %\label{prop:Spec(V) is manifold}
 The space $\Spec(V)$ is a closed submanifold of $N(V)\times \C$ of complex dimension~$n^2$.
\end{Proposition}
\begin{proof}
 Consider the restricted characteristic polynomial $p \from N(V) \times \C \to \C$. As $\Spec(V)$ is the zero set of this map, by the Submersion Theorem it suffices to show that 0 is a regular value. Hence we consider the differential $\dif p(A,z)$, which contains the term $\pder{p}{z}(A,z) \dif z$. This is a~surjection whenever $\pder{p}{z}(A,z)\ne0$, which holds on all of $\Spec(V)$ by Lemma~\ref{lem:non-deg <-> non-zero derivative}.
 Hence $\Spec(V)$ is a closed submanifold, of the same dimension as $N(V)$.
\end{proof}

We are now in place to deduce the bundle structure of $\pi_\lambda$. A fiber is of the form $\Spec(A)$ with $A\in N(V)$, hence by definition a set of $n$ distinct points. We thus take the model fiber to be $I_n$. Inspection of $\pi_\lambda$ yields the following argument.
\begin{Theorem} \label{thm:Spec(V) bundle over N(V)}
 The map $\pi_\lambda \from \Spec(V) \to N(V)$ defines a fiber bundle with model fiber $I_n$.
\end{Theorem}
\begin{proof}
 First, $\pi_\lambda$ is a surjective submersion; the implicit function theorem provides local eigenvalues, and so local sections, through any point of $\Spec(V)$. By dimension count, $\pi_\lambda$ is a local diffeomorphism. As $\pi_\lambda$ is also proper, it is a covering map \cite{Lee2012IntroductionManifolds}. Then, as each fiber has exactly~$n$ elements, $\pi_\lambda$ is an $I_n$-bundle. Indeed, a local trivialization is a map of the form
 \begin{align}
 \phi \from\ U \times I_n &\to \Spec(V)|_U,\nonumber
 \\
 (A,i) &\mapsto (A,\lambda_i(A))\label{eq:local triv Spec(V)}
 \end{align}
 with $\lambda_1,\dots,\lambda_n$ distinct local eigenvalues.
\end{proof}

This result formalizes the idea that $\Spec(V)$ is locally the union of the graphs of $n$ local eigenvalues. We remark that the inverse of the local trivialization $\phi$ in equation~\eqref{eq:local triv Spec(V)} above can be written as $(A,\lambda)\mapsto(A,\#(A,\lambda))$, where $\# \from \Spec(V)|_U \to I_n$ yields the label of the graph in which a point lies. This map $\#$ one can interpret as the ``local labeling'' induced by the chosen local eigenvalues. We will see this map again when discussing the eigenvectors in Section~\ref{sec:eigvec bundle and GPs}.

\subsection{Geometry behind swaps of energies} \label{sec:Spec(H)}

We will now treat how the abstract theory discussed above facilitates the study of instantaneous energies in adiabatic quantum mechanics, including the swaps of energies related to exceptional points (EPs). We assume that instantaneous eigenstates will remain instantaneous eigenstates. Hence instantaneous energies are well-defined. The change of energy in time we will treat formally using covering theory.

To start, we assume that an experimental set-up is captured by a Hamiltonian operator. Typically, this Hamiltonian depends on the available system parameters, e.g., cavity size and field strength in optics and photonics (see, e.g., the review in \cite{Miri2019ExceptionalPhotonics}).
This induces a manifold of system parameter values, which we denote by $M$. We thus obtain a family of experimental set-ups, and so a family of Hamiltonians, described by a map
\begin{equation*}
 H \from\ M \to \End(V),
\end{equation*}
which sends a configuration of system parameters $x\in M$ to the Hamiltonian $H(x)$ corresponding to that configuration. For simplicity, we will assume that $M$ and $H$ are smooth. We will refer to~$H$ as the Hamiltonian family, where each $H(x)$ is a Hamiltonian operator on $V$. We do however not require the Hamiltonian operators to be Hermitian. We refer to the eigenvalues as energies and to the eigenvectors as eigenstates.

The idea is now to vary the system parameters. In practice this means we follow a path $\gamma$ in~$M$, whose initial point $x_0$ serves as a reference. This results in the time-dependent Hamilto\-nian~$H(\gamma(t))$. A state $\psi(t)$ is called an instantaneous eigenstate at time $t$ if it satisfies the eigenvalue problem $H(\gamma(t))\psi(t)=E(t)\psi(t)$, with $E(t)$ an energy of $H(\gamma(t))$. The idea of the adiabatic approximation is that this relation is preserved in time, at least approximately. However, this means one first has to make sure that the function $E(t)$ is well-defined.

This fundamental fact can now easily be deduced from the covering properties of $\Spec(V)$. First, we remark that in order to unambiguously follow a specific energy level of $H(\gamma(t))$, the energy level may not become degenerate at any time. Hence we require $H(\gamma(t))$ to be non-degenerate for all $t$, i.e., $t\mapsto H(\gamma(t))$ is a path in $N(V)$. Let $E_0$ be the energy level of the initial Hamiltonian $H(x_0)$ that we wish to follow in time. Clearly, this defines the point $(H(x_0),E_0)$ in the fiber of $\Spec(V)$ above $H(x_0)$. Hence, it specifies a unique lift of $H\circ \gamma$ to $\Spec(V)$, which is of the form $(H(\gamma(t)),E(t))$ for some function $E=E(t)$. In this way, we obtain the instantaneous energy for all relevant times in a formal way.

The intuition of this lifting argument is to follow the initial energy along the energy bands of the family $H$. If $H$ is fixed, it is convenient to consider these energy bands directly. This can be done by taking the pull-back of the bundle $\Spec(V) \to N(V)$ along $H$. To obtain this pull-back, we must restrict the system parameters to those were $H$ is non-degenerate, i.e., we must restrict $M$ to the subspace
\begin{equation*}
 N(H):=H^{-1}(N(V))=\set{x\in M}{H(x) \text{ has non-degenerate energy levels}}.
\end{equation*}
The energy bands of $H$, with the degenerate points omitted, then form the space
\begin{equation*}
 \Spec(H):=\set{(x,E) \in N(H)\times\C}{E\in\Spec(H(x))},
\end{equation*}
i.e., the space of all pairs of a specific configuration $x$ of the system parameters and an energy~$E$ of the system for this configuration. We call $\Spec(H)$ the \emph{spectrum bundle of $H$}; it has a~natural projection $\pi_\lambda^H \from (x,E)\mapsto x$ to $N(H)$, so that the fiber above $x$ is simply $\Spec(H(x))$. The bundle property of $\pi_\lambda^H$ is guaranteed by the pull-back construction, and relates to $\pi_\lambda$ as given by the pull-back diagram below:
\begin{equation*}
 \begin{tikzcd}
 \Spec(H) \ar{r} \ar{d}{\pi_\lambda^H} & \Spec(V) \ar{d}{\pi_\lambda} && (x,E) \ar[mapsto]{r}\ar[mapsto]{d} & (H(x),E) \ar[mapsto]{d}\\
 N(H) \ar{r}{H} & N(V), && x \ar[mapsto]{r} & H(x).
 \end{tikzcd}
\end{equation*}
We see that the path $\gamma$ can be lifted directly to $\Spec(H)$, which is then both an intuitive and formally correct way to obtain the function $E(t)$ directly. Indeed, by assumption $\gamma$ lies in $N(H)$, the initial point is now $(x_0,E_0)$, and the lift is of the form $(\gamma(t),E(t))$.

We see that the study of $\Spec(H)$ is a straightforward generalization of Kato's original setting~\cite{Kato1966PerturbationOperators}. There, $H$ was assumed to depend analytically on a single complex variable, so that the energies were locally given by analytic functions. By analytically continuing them along a~path around a degeneracy, the eigenvalues could return in a different order, and such a degeneracy was called an EP. Here, we see that $\Spec(H)$ allows us to drop the assumption of analytic depen\-dence; instead of analytic continuation we can use lifting along a covering map. Hence we broadened the perspective from complex analytic to general continuous dependence. Note that this is necessary to include Hermitian or $\PT$-symmetric Hamiltonians, as these are defined by a conjugate-linear operation and hence do not fit in an analytic setting.

The study of the permutations of energies in the family $H$ is in this construction tantamount to the study of the covering properties of $\Spec(H)$. In fact, the permutations of energies of $H(x_0)$ are naturally described by the monodromy action at $x_0$, as we will argue now.
First, as we wish to compare energies, we should restore the original system configuration at the end (see also \cite{Pap2018Non-AbelianPoints} for a practical argument for this). That is, the path $\gamma$ in $N(H)$ should return to our reference $x_0\in N(H)$. In~other words, $\gamma$ should be a loop based at $x_0$. Let us write $\Loop(N(H),x_0)$ for the set of such loops. We are then interested in how the energies of~$H(x_0)$ return upon following them along $\gamma$. For an energy $E\in \Spec(H(x_0))$, this is given by the lift of $\gamma$ to $\Spec(H)$; the final energy is the endpoint of the lift, which we write as $p_\gamma(E)$. Doing this for every energy of~$H(x_0)$ yields the map $p_\gamma \from \Spec(H(x_0)) \to \Spec(H(x_0))$. This map is a bijection/permutation of $\Spec(H(x_0))$; its inverse is given by traversing $\gamma$ in the opposite direction. Going over all loops in $\Loop(N(H),x_0)$ then yields the group of all possible permutations of $\Spec(H(x_0))$, namely the group
\begin{equation*}
 \Hol^{\Spec(H)}_{N(H)}(x_0)=\set{p_\gamma\in \Aut(\Spec(H(x_0)))}{\gamma\in \Loop(N(H),x_0)},
\end{equation*}
where the automorphisms are simply bijections. This is indeed a holonomy group as lifting along a covering can be regarded as parallel transport. We remark that this group was called $\Lambda(x_0)$ in \cite{Pap2018Non-AbelianPoints}, which is a straightforward generalization of the permutation group studied by Kato \cite{Kato1966PerturbationOperators} from complex analytic to continuous Hamiltonian families.

\begin{Example} \label{ex:start EP2}
 Let us consider a standard case, namely a family with EP2, where an EP2 is an~EP permuting 2 energies. Explicitly, consider $H\from \C \to \End\big(\C^2\big)$ given by
 \begin{equation*}
 H(x)=
 \begin{pmatrix}
 1&x\\
 x&-1
 \end{pmatrix}\!.
 \end{equation*}
 The eigenvalues are given by the multi-valued function $\sqrt{1+x^2}$, so $N(H)=\C\setminus\{\pm {\rm i}\}$ and $\Spec(H)$ is the graph of this multi-valued function, which in this case is a Riemann surface. Let us take a branch cut, and write $E_\pm(x)$ for the two energy branches. Clearly, $\Spec(H(x_0))=\{E_+(x_0),E_-(x_0)\}$, and if $\gamma$ is a loop based at $x_0$ encircling an EP once, it induces the map
 \begin{equation*}
 p_\gamma \from\ E_\pm(x_0) \mapsto E_\mp(x_0).
 \end{equation*}
 Of course, if $\gamma$ does not encircle any EP, we obtain $p_\gamma=\id_{\Spec(H(x_0))}$.
\end{Example}

The association $\gamma \mapsto p_\gamma$ brings us to a monodromy action as follows. The map $p_\gamma$ is invariant under continuous deformation of $\gamma$, i.e., it only depends on its homotopy class $[\gamma]$. The set of such homotopy classes of loops based at $x_0$ is the (based) fundamental group $\pi_1(N(H),x_0)$. The exchange of energies can then be expressed by the homomorphism $\pi_1(N(H),x_0)\to \Hol^{\Spec(H)}_{N(H)}(x_0)$, $[\gamma] \mapsto p_\gamma$, or equivalently as an action of $\pi_1(N(H),x_0)$ on $\Spec(H(x_0))$. This is known as the monodromy action, which we write as
\begin{align*}
 \pi_1(N(H),x_0) \times \Spec(H(x_0)) &\to \Spec(H(x_0)),
 \\
 [\gamma]\cdot (x_0,E)&=(x_0,p_\gamma(E)).
 \end{align*}
A picture to have in mind is Figure~\ref{fig:swapping_lift}; following the energy level along $\gamma$ defines a path in $\Spec(H)$, which may end at a different energy. An explicit example we treat in Example~\ref{exmp:monodromy action EP2} below. We~note that the group $\pi_1(N(H),x_0)$ is typically non-trivial. Indeed, by Lemma~\ref{lem:props N(V)} the space $N(V)$ has real codimension 2 in $\End(V)$ and so a non-trivial fundamental group, which then holds for a generic $N(H)$ as well. We remark that the relevance of the monodromy action for EPs was reported in \cite{Tanaka2015BlochHolonomy,Tanaka2017PathEvolution}. The space used there is either $\Spec(H)$ or a related space we obtain in Section~\ref{sec:holonomy}.

\begin{figure}[h]
 \centering
 \includegraphics[scale=.28]{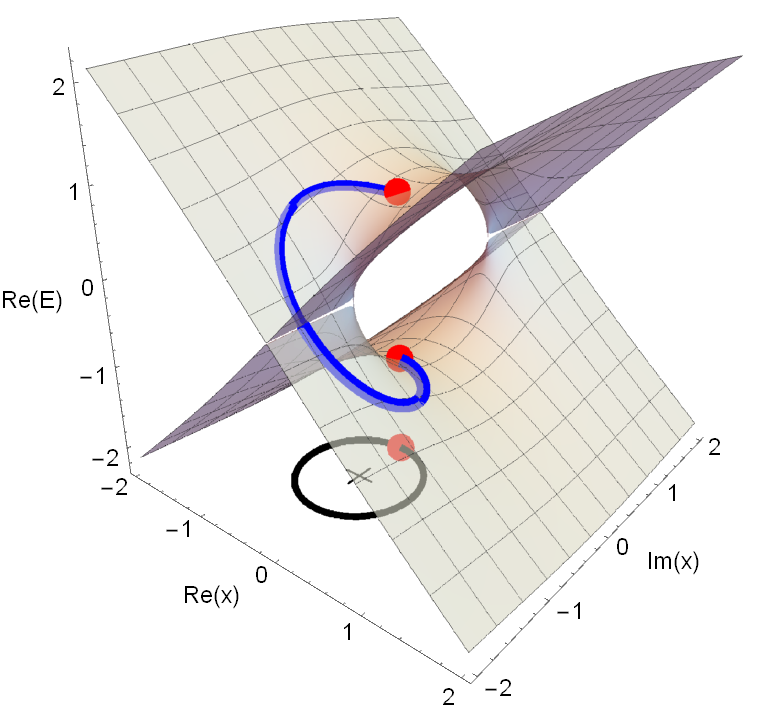}
 \caption{Exchange of energies around an EP is given by the monodromy of the chosen path. Plotted is the real part of the energy levels of the Hamiltonian family in Example~\ref{ex:start EP2}. Upon changing system parameters (black circle) from the reference value (red dot below), the energy moves (blue line) along the sheets and does not return to itself. Such behavior only happens around special degeneracies, namely the EPs. Here the black $\mathbf{+}$ marks the EP of the system related to the drawn exchange.}
 \label{fig:swapping_lift}
\end{figure}

\begin{Example} \label{exmp:monodromy action EP2}
 Continuing Example~\ref{ex:start EP2}, it follows that $\pi_1(N(H),x_0)$ is isomorphic to the fundamental group of the ``figure 8'', which is a free group on two generators, for any reference $x_0$. Picking $x_0=0$, the monodromy action maps both generators of $\pi_1(N(H),0)$ to the interchange of $+1$ and $-1$ in $\Spec(H(0))$.
\end{Example}

The realization of $\Spec(H)$ as a pull-back of $\Spec(V)$ allows us to distinguish between inherent geometric properties coming from $\Spec(V)$ versus artifacts of the family $H$. Such artifacts can result in different, i.e., non-isomorphic, spaces $\Spec(H)$ for the different $H$. For example, we can explicitly add an artifact to a particular $H\from M\to \End(V)$ by introducing a fictitious system parameter, i.e., an artificial variable that does not correspond to any change in the experimental set-up. This increases the dimensions of both $M$ and $\Spec(H)$, but there is no additional physical information. Another example is to include multiple degenerate levels in $H$ which do not depend on the system parameters. In this case $N(H)$ becomes empty. Of course, these examples are highly artificial, yet we see it is advisable to be careful not to include any non-physical information in $H$.

\subsubsection{Merging path method} \label{sec:merging path method}

A common way to demonstrate the existence of an EP in an experimental set-up is to follow a~loop in parameter space and trace the eigenvalues accordingly. If the eigenvalues do not return to themselves, then there must be (at least one) EP structure inside the loop, provided that any path in the parameter space $M$ can be contracted to a point. Such an argument can already be found in \cite{Heiss1999PhasesRepulsion}. We observe that this method does not attribute a swap to a particular EP~-- these are not even present in $N(H)$~-- but rather to a homotopy class of loops, which in turn signals the presence of a degeneracy structure. In practice, the traversing of a loop is done in a~stroboscopic way by measuring the energy for sufficiently close discrete parameter values (e.g., in \cite{Xu2016TopologicalPoints}). The formal mathematics behind this argument is straightforward.

Let us go through this merging path method step by step. The change of system parameters is given by a loop $\gamma$ in $N(H)$. The spectrum then changes according to the loop $\Spec(H(\gamma(t)))$ in $\binom{\C}{n}$, with $n$ the dimension of the state space. To extract the permutation of the energies, note that we cannot work in $\binom{\C}{n}$ alone as the (unordered) spectrum will return to itself regardless of the presence of EPs. We must thus lift the loop $\Spec(H(\gamma(t)))$ to $\Cun$ in order to observe and extract a permutation. This argument translates directly into the diagram
\begin{equation} \label{eq:merging path diagram}
 \begin{tikzcd}
 &&\Cun \ar{d}\\
 N(H) \ar[dashed]{r}{H}& N(V) \ar{r}{\Spec} & \binom{\C}{n},
 \end{tikzcd}
\end{equation}
which is the bottom right corner of diagram~\eqref{eq:summarizing diagram} plus the adaptation to the specific system defined by~$H$.

\subsubsection{Real eigenvalue case}

The space $\Spec(H)$ can have additional properties if the family $H$ satisfies additional assumptions. One interesting assumption on $H$ is that $H(x)$ has real energies for all $x\in M$. This occurs, e.g., when $H(x)$ is always Hermitian w.r.t.\ some given inner product on $V$, or when all~$H(x)$ have exact $\PT$ symmetry (see \cite{Mostafazadeh2003ExactHermiticity} on the relation between these). In this case, $\Spec(H)$ is trivial over $N(H)$, which means that all energy bands are globally disconnected. In particular, no exchanges can occur and EPs cannot be present. We first prove this on the abstract level, and then use the pull-back by $H$ to go to $\Spec(H)$. We remark that the argument also holds if the eigenvalues are taken in another totally ordered subset of $\C$.
\begin{Proposition}
 Write $R(V)$ for the subspace of $N(V)$ of matrices with real eigenvalues. The bundle $\Spec(V)$ restricted to $R(V)$ is trivial.
\end{Proposition}
\begin{proof}
 If $A \in R(V)$, then there is a unique ordering of $\Spec(A)$ by indexing the eigenvalues from lowest to highest $\lambda_1(A) < \dots <\lambda_n(A)$. This establishes a global labelling map $\Spec(V)|_{R(V)} \to R(V) \times I_n$ defined as $(A,\lambda_i) \mapsto (A,i)$, which is a homeomorphism.
\end{proof}
\begin{Corollary} \label{cor:real eigenvalues->trivial}
 If $H(x)$ has real eigenvalues for all $x\in N(H)$, then $\Spec(H)$ is a trivial bundle. In particular, there are $n$ distinct energy bands $E_1(x),\dots,E_n(x)$ defined continuously on all of~$N(H)$, and the family $H$ has no EPs.
\end{Corollary}

\section{Eigenvector bundle and geometric phases} \label{sec:eigvec bundle and GPs}

In the previous section we showed how the energies of an adiabatic quantum system, in particular their exchanges around EPs, can be treated using the covering space $\Spec(V)$. We will now extend this formalism to eigenstates, where again non-Hermitian Hamiltonians are allowed.

\subsection{The eigenvector bundle} %\label{sec:eigenvector bundle}

Let us start by studying how eigenvectors vary with the operator. We will use the following notation. The set of all eigenvectors of an operator $A\in \End(V)$ corresponding to an eigenvalue $\lambda \in \Spec(A)$ we denote as
\begin{equation*}
 \Eig_\lambda(A)=\set{v\in V\setminus\{0\}}{Av=\lambda v}=\ker(A-\lambda I)\setminus\{0\}.
\end{equation*}
We call this the space of eigenvectors, or eigenvector space, of $A$ corresponding to $\lambda$. This terminology emphasizes that the zero vector is excluded. In this way, each $\Eig_\lambda(A)$ is a free $\C^\times$-manifold; any non-zero multiple of an eigenvector is again an eigenvector. Similarly, we define the eigenvector space $\Eig(A)$ of $A$ to be the set of all eigenvectors of $A$. Naturally,
\begin{equation*} %\label{eq:Eig(A) decomp}
 \Eig(A)=\bigsqcup_{\lambda\in \Spec(A)} \Eig_\lambda(A),
\end{equation*}
which is immediately the partition of $\Eig(A)$ into its connected components. We see that $\Eig(A)$ is also endowed with $\C^\times$-scaling, but need not be a manifold as the eigenvector spaces of the eigenvalues need not have the same dimension. However, in case $A\in N(V)$ clearly $\Eig(A)$ is a~$\C^\times$-manifold of complex dimension 1, and the above decomposition shows $\Eig(A)$ as a union of $\C^\times$-torsors.

Here we already found a hint that also concerning eigenvectors we need to restrict to non-degenerate operators. Indeed, if we wish to view eigenvectors varying with the operator as local sections of a bundle, i.e., locally defined functions $v(A)$ so that $v(A)\in \Eig(A)$, then degenerate operators will again give rise to singularities. We encounter the same situation as before, be it with more facets. The problem with eigenvalues was that the number of distinct ones suddenly drops at a degenerate operator. Here, this implies that $\Eig(A)$ suddenly consists of less than~$n$ connected components, which renders a bundle structure impossible. In addition, we see for degenerate yet diagonalizable operators that the eigenvector space can consist of parts with different dimension, hence $\Eig(A)$ is not even a manifold for such operators. Hence we again restrict to non-degenerate operators. The space that we obtain this way we state in a definition because of its importance.
\begin{Definition}[eigenvector bundle]
 Given the vector space $V$, define its \emph{eigenvector bundle} to be the space
 \begin{equation*}
 \Eig(V)=\set{(A,\lambda,v)\in \Spec(V) \times V\setminus\{0\}}{v\in\Eig_\lambda(A)}.
 \end{equation*}
\end{Definition}
Again, we have used the term bundle straightaway, and in the remainder of this section we will justify this term.

In fact, we will find that $\Eig(V)$ is a bundle with respect to two different projections. The idea is that we can view $\Eig(V)$ as a union of total eigenvector spaces or as a disjoint union of individual eigenrays;
\begin{equation*}
 \Eig(V)=\bigsqcup_{A\in N(V)}\Eig(A)=\bigsqcup_{(A,\lambda)\in \Spec(V)}\Eig_\lambda(A).
\end{equation*}
Both of these viewpoints come with their own projection. For the first, one projects \mbox{$(A,\lambda,v) \!\mapsto\! A$}, which we write as the map $\pi_{\lambda v} \from \Eig(V) \to N(V)$. For the second,
one only omits the eigenvector and projects to $(A,\lambda)$, which defines a map $\pi_v \from \Eig(V)\to \Spec(V)$. These two projections are related by the projection $\pi_\lambda$, as summarized in the diagram below:
\begin{equation} \label{eq:Eig(V) Spec(V) N(V) triangle}
 \begin{tikzcd}[column sep=1em]
 \Eig(V) \ar{rr}{\pi_v} \ar[swap]{rd}{\pi_{\lambda v}} && \Spec(V) \ar{ld}{\pi_\lambda} && (A,\lambda,v)\ar[mapsto]{rr} \ar[mapsto,swap]{rd} && (A,\lambda) \ar[mapsto]{ld}\\
 &N(V), &&&& A.
 \end{tikzcd}
\end{equation}

Our first step in proving the bundle claims is to show that $\Eig(V)$ is a smooth manifold. Clearly $\Eig(V)$ is the zero set of
\begin{align*} %\label{eq:def P}
 P \from\ \Spec(V) \times V \setminus \{0\} &\to V,
 \\
 (A,\lambda,v) &\mapsto (\lambda I-A)v.
 \end{align*}
In this way we obtain the following structure on $\Eig(V)$.

\begin{Proposition}
 The space $\Eig(V)$ is a closed algebraic $\C^\times$-submanifold of $\Spec(V) \times V\setminus\{0\}$ of complex dimension $n^2+1$.
Moreover, the projection maps $\pi_{\lambda v}$ and $\pi_v$ are smooth.
\end{Proposition}

\begin{proof}
 Pick a point $(A_0,\lambda_0,v_0)\in \Eig(V)$ and pick local coordinates $(A,\lambda)$ on $\Spec(V)$ around $(A_0,\lambda_0)$; this implicitly defines $\lambda=\lambda(A)$ with $\lambda(A_0)=\lambda_0$. Consider now the differential
 \begin{equation*} %\label{eq:dP}
 \dif P=(\lambda(A)I-A)\dif v+\dif(\lambda(A)I-A)v.
 \end{equation*}
 As $P$ is constant on the $\C^\times$-orbits, $\dif P$ has rank at most $n-1$. On the other hand, the term $(\lambda(A)I-A)\dif v$ has the same rank as $\lambda(A)I-A$, which is $n-1$ by non-degeneracy. Hence $\dif P$ has constant rank $n-1$, and so by the constant rank theorem \cite{Lee2012IntroductionManifolds} the fibers of $P$ are closed submanifolds of (complex) dimension $\big(n^2+n\big)-(n-1)=n^2+1$. The maps $\pi_{\lambda v}$ and $\pi_v$ are then restrictions of smooth projections to a submanifold, hence smooth.
\end{proof}

Both $\pi_{\lambda v}$ and $\pi_v$ are invariant w.r.t.\ the $\C^\times$-action on $\Eig(V)$, which hints at a principal bundle structure. For $\pi_v$, this holds; its fibers are eigenrays $\Eig_\lambda(A)$, which are $\C^\times$-torsors, and~$\pi_v$ coincides with the quotient of the action on $\Eig(V)$. However, $\pi_{\lambda v}$ is clearly not principal. The fiber of $\pi_{\lambda v}$ above $A\in N(V)$ is the total eigenvector space $\Eig(A)$ of $A$, which is a union of $\C^\times$-torsors, hence not a $\C^\times$-torsor itself for $n>1$. Still, with the eye on adiabatic quantum mechanics, we are motivated to describe $\pi_{\lambda v}$ in more detail.

We thus turn to a wider class of bundles, which extends the principal bundles. In short, principal $G$-bundles are bundles endowed with a $G$-action so that the fibers are $G$-torsors (hence diffeomorphic to $G$) and the projection admits $G$-equivariant local trivializations. Let us only change the model fiber; we allow it to be a semi-torsor, i.e., a disjoint union of torsors (hence diffeomorphic to $G\times I$ for some index set $I$). Bundles satisfying this more general condition we call semi-principal $G$-bundles \cite{Pap2020FramesTheory}. Clearly, as any torsor is a semi-torsor, this is an extension of the principal bundles. In case the semi-torsor consists of $n$ torsors, where $n$ is finite, we write the model fiber as $G\times I_n$. More details on semi-principal bundles can be found in \cite{Pap2020FramesTheory}.

Let us show that $\pi_{\lambda v}\from \Eig(V) \to N(V)$ is a semi-principal $\C^\times$-bundle with model fiber $\C^\times \times I_n$. The main property left to show is the equivariant local triviality, for which we prepare ourselves with the following extension property. We use the language of smooth maps, but they may be taken to be algebraic.
\begin{Lemma} \label{lem:local eigenvector}
 For every point $A\in N(V)$, there is an open neighborhood $U\subset N(V)$ of $A$ and~$n$ eigenvector functions $v_1,\dots,v_n \from U \to V$ which constitute an eigenframe at each point of $U$. Moreover, this moving eigenframe can be chosen to extend any eigenframe of $A$.
\end{Lemma}
\begin{proof}
 The local eigenvector functions exist by Corollary~\ref{cor:local eigvals eigvecs}. The eigenframe they form at~$A$ can be changed to any given eigenframe of $A$ by acting with the appropriate group element in~$\struc$, as we did in the proof of Lemma~\ref{lem:Xi is principal}.
\end{proof}

\begin{Theorem} \label{thm:Eig(V) bundle over N(V)}
 The map $\pi_{\lambda v} \from \Eig(V) \to N(V)$ defines a semi-principal $\C^\times$-bundle with model fiber $\C^\times \times I_n$.
\end{Theorem}
\begin{proof}
 Pick $A_0\in N(V)$, let $U\subset N(V)$ be a neighborhood of $A_0$ and for $i=1,\dots,n$ let $v_i \from U \to V$ be a local eigenvector corresponding to eigenvalue $\lambda_i(A_0)$, according to Lemma~\ref{lem:local eigenvector}. One has a map
 \begin{align*}
 \Phi \from\ U \times I_n \times \C^\times &\to \Eig(V)|_U,
 \\
 (A,i,z) &\mapsto (A,\lambda_i(A),zv_i(A)),
 \end{align*}
 which is smooth and clearly intertwines the $\C^\times$-actions. Its inverse is
 \begin{equation*}
 (A,\lambda,v) \mapsto \big(A,\#(A,\lambda),\big[v/v_{\#(A,\lambda)}(A)\big]\big),
 \end{equation*}
 which satisfies the same properties. Here $\#$ is the local labelling, and the last entry denotes the scalar $z$ such that $v=zv_{\#(A,\lambda)}(A)$. It follows that $\Phi$ is a $\C^\times$-equivariant local trivialization, hence the bundle property follows.
\end{proof}

We observe that $\pi_\lambda$ and $\pi_v$ are closely tied to $\pi_{\lambda v}$. Namely, it holds that every semi-principal $G$-bundle decomposes naturally into a principal $G$-bundle and a covering space \cite{Pap2020FramesTheory}. These maps form exactly such a decomposition.
\begin{Proposition} %\label{prop:Eig(V) is decomp}
 The rule $\pi_{\lambda v}=\pi_\lambda \circ \pi_v$ is the decomposition of $\pi_{\lambda v}$ in a principal bundle projection followed by a covering map.
\end{Proposition}

\subsection{Connection on the eigenvector bundle}

Our next step is to show that $\Eig(V)$ supports a canonical connection. This connection is compatible with both the semi-principal projection $\pi_{\lambda v}$ and the principal projection $\pi_v$. Hence, we prefer to avoid the term ``principal connection'', and make reference only to the action instead. We will thus use the concept of a $G$-connection defined as follows. Let $G$ be a Lie group with algebra $\g$; a 1-form $\omega$ on a $G$-manifold $M$ is a $G$-connection if $\omega_m\circ a_m=\id_\g$ for all $m\in M$, with $a_m$ the infinitesimal action at $m$, and $L_g^*\omega=\Ad_g(\omega)$ for any $g\in G$. See also \cite{Pap2020FramesTheory} on this.

To describe the connection 1-form on $\Eig(V)$, it is helpful to extend our notation. The set of eigencovectors of $A\in \End(V)$ we denote as $\Eig^\vee(A)$, and those corresponding to a specific eigenvalue $\lambda$ by $\Eig^\vee_\lambda(A)$. Note that the zero covector is excluded. As we saw in Lemma~\ref{lem:eigenframe}, for diagonalizable $A$, there is a bijection between eigenframes and eigencoframes. For non-degenerate $A$ a stronger statement holds, namely, the bijection already appears on the level of individual eigenvectors and eigencovectors.
\begin{Lemma}
 For a non-degenerate matrix $A$, there is a canonical bijection $\Eig(A) \to \Eig^\vee(A)$ sending $v\in \Eig_\lambda(A)$ to the unique covector $\theta \in \Eig^\vee_\lambda(A)$ defined by $\theta(v)=1$.
\end{Lemma}
\begin{proof}
 Following the decomposition of $V$ by the eigenspaces of $A$, define the covector $\theta$ to vanish on every eigenspace except the one of $\lambda$, where it is fixed by setting $\theta(v)=1$.
\end{proof}

This correspondence shows that the scaling action on $\Eig(A)$ corresponds to inverse scaling on $\Eig^\vee(A)$, i.e., we obtain the $\C^\times$-action on $\Eig^\vee(A)$ given by
\begin{equation*}
 z \cdot \theta =\theta z^{-1}.
\end{equation*}
The bijection $\Eig(A) \to \Eig^\vee(A)$ is then an isomorphism of $\C^\times$-manifolds. In addition, as any triple $(A,\lambda,v)\in \Eig(V)$ defines a unique $\theta$, we are free to augment the triple to $(A,\lambda,v,\theta)$. This will be a convenient notation when defining maps on $\Eig(V)$, such as the connection 1-form we seek at the moment. Similarly, one is free to omit $\lambda$ from the triple. We will use these alternative notations frequently.

We are now set to describe the connection 1-form on $\Eig(V)$. By definition, this 1-form is a~right-inverse of the infinitesimal action, which is specified by the vector field
\begin{equation*}
 \pder{}{v}\bigg|_{(A,v)}:=\dtzero \exp(t)\cdot(A,v)=\dtzero \big(A,{\rm e}^{t}v\big),
\end{equation*}
which can be thought of as a unit vector field pointing along the eigenray of $A$ through $v$. Hence, given a path $\Gamma(t)=(A(t),\lambda(t),v(t),\theta(t))$ in $\Eig(V)$, the connection should measure the component of $\dot{v}$ along $v$. This cannot be done with the vector space structure of $V$ alone, but given the operator this is possible. Indeed, the eigenrays of $A(t)$ provide a natural decomposition of $V$ at time $t$. According to this choice, the coefficient of $\dot{v}$ along $v$ is simply $\theta(\dot{v})$.

Let us describe this argument formally. To obtain the quantity $\dot{v}$ from $\dot{\Gamma}$, we can use the differential of the projection $\pr_V \from \Eig(V) \to V$, $(A,\lambda,v) \mapsto v$. Technically, this will take values in the tangent space $T_{(A,v)}V$, but this can be identified with $V$ in a canonical way. Let us write $\dif v$ for $(\pr_V)_*$ with the image viewed in $V$. This map can be followed up with $\theta$, which yields the connection 1-form.

\begin{Proposition} \label{prop:omega Ctimes connection}
 The $1$-form $\omega \in \Omega^1(\Eig(V),\C)$ given by
 \begin{equation*}
 \omega=\theta \dif v
 \end{equation*}
 is a $\C^\times$-connection.
\end{Proposition}

\begin{proof}
 As $\C^\times$ is commutative, the equivariance of $\omega$ reduces to invariance. This holds by the opposite scaling of $\theta$ and $v$: for $z\in \C^\times$, $(L_z)^*\omega=\big(\theta z^{-1}\big)(\dif(zv))=\theta \dif v=\omega$.
 The left-inverse property of $\omega$ follows as $\omega(\partial_v)=\theta \dif v(\partial_v)=\theta(v)=1$.
\end{proof}

For future reference, we also remark that this induces the global curvature form $K$ on $\Eig(V)$ in the standard way, i.e., $K:=\dif \omega +\frac{1}{2}[\omega,\omega]=\dif \omega$. In coordinates it reads $K=\dif \theta \wedge \dif v$, where~$\dif \theta$ is defined similar to $\dif v$. As $\C^\times$ is commutative, the curvature $K$ is invariant under the group action. Hence it admits push-forward along the quotient map $\pi_v$ to $\Spec(V)$, resulting in the following.
\begin{Lemma}
 The curvature form $K$ on $\Eig(V)$ reduces to a unique $2$-form $k\in \Omega^2(\Spec(V),\C)$, which satisfies $K=\pi_v^*(k)$.
\end{Lemma}

\subsection{Geometry behind the geometric phase} %\label{sec:Eig(H)}

We will now show that the connection $\omega$ yields the geometric phases in adiabatic quantum mechanics in a natural way. That is, $\Eig(V)$ provides a geometric model for the geometric phase. This is also applicable to non-cyclic states, which appear in the presence of EPs of non-Hermitian Hamiltonian families.

An argument similar to the one that led us from $\Spec(V)$ to $\Spec(H)$ in Section~\ref{sec:Spec(H)} holds here. Namely, the space $\Eig(V)$ can be seen as a general abstract model, which can be adjusted for a specific Hamiltonian family using pull-back by $H$. This pull-back of $\Eig(V)$ along the Hamiltonian family $H$ yields the \emph{eigenstate bundle of $H$}, given by
\begin{equation*}
 \Eig(H)=\set{(x,E,\psi) \in \Spec(H)\times V\setminus\{0\}}{\psi\in\Eig_E(H(x))}.
\end{equation*}
We observe that the pull-back construction of this space immediately guarantees various properties of $\Eig(H)$. Clearly, $\Eig(H)$ is a smooth manifold. One can also write an element as $(x,E,\psi,\chi)$, where $\chi$ is the unique left-eigenstate of $H(x)$ corresponding to $E$ so that $\chi(\psi)=1$. Similarly, one may omit $E$ from the tuple. The projection $\pi_{\lambda v}^H \from \Eig(H)\to N(H)$ is a semi-principal $\C^\times$-bundle, and the projection $\pi_v^H \from \Eig(H)\to \Spec(H)$ is a principal $\C^\times$-bundle. The latter bundle confirms that also $\Eig(H)$ can have non-trivial topology arising from the permutations around the EPs of $H$.

The local sections of $\pi_{\lambda v}^H$ are in correspondence with local eigenstates. Here, a local eigenstate is a smooth function $\psi \from U\to \Eig(H)$, with $U\subset N(H)$ open, such that $\psi(x)$ is an eigenstate of~$H(x)$. Clearly, this fixes a local eigenvalue $E(x)$. This data is summarized in the corresponding local section $\hat{\psi}$ as
\begin{equation*}
 \hat{\psi} \from\ x\mapsto (x,E(x),\psi(x)).
\end{equation*}
Observe that for each $x\in U$, knowing $E(x)$ and $\psi(x)$ fixes a unique left-eigenstate $\chi(x)$ of $H(x)$, which depends smoothly on $x$. Hence we can include $\chi(x)$ as the fourth component of the local section $\hat{\psi}$, following our earlier remark on notation of elements of $\Eig(H)$.

The space $\Eig(H)$ also inherits a connection, which we write as $\omega_H$. Again, for any family $H$, the smoothness and other properties of $\Eig(H)$ are immediate from the pull-back construction. It is now easy to verify that parallel transport w.r.t.\ this connection is equivalent to studying the geometric phase. Namely, using the local section $\hat{\psi}$, we can express $\omega_H$ locally as
\begin{equation*}
 \hat{\psi}^*\omega_H=\chi \dif \psi.
\end{equation*}
This is indeed the expression reported by Garrison and Wright \cite{Garrison1988ComplexSystems} as the generalization of the Berry connection. That is, the parallel transport on $\Eig(H)$ defined via the horizontal lifts w.r.t.~$\omega_H$ is equivalent to the calculation of geometric phase, for both Hermitian and non-Hermitian Hamiltonians. We may thus study geometric phases by studying the space $\Eig(H)$.

For example, the connection $\omega_H$ on $\Eig(H)$ can be flat, in which case the geometric phase is actually of a topological nature. By this we mean that the geometric phase due to a loop $\gamma$ is invariant under continuous deformation of $\gamma$, which is equivalent to vanishing of the curvature~$K_H$ of $\Eig(H)$. Hence, only non-contractible loops can yield a non-trivial geometric phase. Flatness is also of practical relevance, as the following result shows.

\begin{Proposition} \label{prop:flat connection}
 The connection $\omega_H$ on $\Eig(H)$ is flat in case
 \begin{itemize}\itemsep=0pt
 \item the Hamiltonian family $H$ is an analytic function in a single complex variable $z$.
 \item $H(x)$ is a symmetric matrix w.r.t.\ a fixed basis of $V$ for all $x \in N(H)$. Explicitly, if a~local eigenstate $\psi$ is such that $\psi^\mathrm{T}\psi=1$, then $\hat{\psi}^*\omega_H=0$.
 \end{itemize}
\end{Proposition}

\begin{proof}
 If $H=H(z)$ is analytic, we can find an analytic local eigenstate $\psi(z)$. Consequently, $\hat{\psi}^*\omega_H=f(z)\dif z$ for some analytic function $f$. Using coordinates $z$ and $\bar{z}$ for the parameter space, $\hat{\psi}^*K_H$ is proportional to $\frac{\partial f(z)}{\partial \bar{z}}$. As this derivative vanishes, it follows that $K_H=0$.

 In the symmetric case, $\psi^\mathrm{T}$ is a non-zero multiple of the left-eigenstate $\chi$ corresponding to $\psi$. Hence the function $\psi^\mathrm{T} \psi$ is non-vanishing, and one may always (locally) scale $\psi$ so that $\psi^\mathrm{T}\psi=1$, implying $\chi=\psi^\mathrm{T}$. By partial integration, $\hat{\psi}^*\omega_H=\tfrac{1}{2}(\chi \dif \psi - \dif \chi \psi)=\tfrac{1}{2}\big(\psi^\mathrm{T} \dif \psi - \dif \psi^\mathrm{T} \psi\big)=0.$
\end{proof}

\subsubsection{Calculating the lift via an ansatz}

An explicit calculation of the lift of a path to $\Eig(H)$ can be done using an ansatz technique. For this technique, it does not matter if the path lies in $N(H)$ or $\Spec(H)$. Indeed, if a path $\gamma$ in $N(H)$ and an initial eigenstate $\psi_0$ of $H(x_0)$ is chosen, then $\psi_0$ has an energy $E_0$, and covering theory yields the corresponding path $\bar{\gamma}(t)=(\gamma(t),E(t))$ in $\Spec(H)$. Conversely, any path $\bar{\gamma}$ in $\Spec(H)$ yields a path in $N(H)$, which obviously lifts to $\bar{\gamma}$. If we restrict to loops the story changes, as we will see after discussing the ansatz technique.

The ansatz technique here is simply the application of the usual one to the parallel transport equation. Let us continue with the path $\gamma$ in $N(V)$ and the initial state $\psi_0$. This data fixes a~unique lift to $\Eig(H)$, which we denote by $\Gamma$ and which is given by
\begin{equation*}
 \Gamma(t)=(\gamma(t),E(t),\Psi(t)),
\end{equation*}
where $\Psi(t)$ is the adiabatically evolved eigenstate at time $t$, up to dynamical phase. Naturally, this lift is also given as $(\bar{\gamma}(t),\Psi(t))$ in the perspective of lifting a path $\bar{\gamma}$ from $\Spec(V)$ to $\Eig(V)$.

In practice, one would not calculate $\Psi(t)$ directly, but instead approach this problem using reference instantaneous eigenstates $\psi(t)$. That is, for each $t$, one finds an eigenstate $\psi(t)$ of~$H(t)$ with energy $E(t)$, or $\psi(t)\in \Eig_{E(t)}(H(t))$ for short. For example, $\psi(t)$ can be obtained by explicitly calculating an eigenstate of $H(\gamma(t))$, and then letting $t$ vary. We may assume that $t\mapsto \psi(t)$ is differentiable and satisfies the initial condition $\psi(0)=\psi_0$. These instantaneous eigenstates define a path $\Gamma_0(t)=(\gamma(t),E(t),\psi(t),\chi(t))$ in $\Eig(H)$, where $\chi(t)$ is the path of accompanying left-eigenstates.

Of course, $\Gamma_0$ does not need to be the lift $\Gamma$, as equivalently $\psi(t)$ does not need to be equal the actual state $\Psi(t)$. However, $\Gamma_0$ can function as an ansatz to calculate $\Gamma$ explicitly. Indeed, both $\psi(t)$ and $\Psi(t)$ must lie in $\Eig_{E(t)}(H(t))$, and so differ only by a scale factor ${\rm e}^{f(t)}$. That is, one has
\begin{equation*}
 \Gamma(t)={\rm e}^{f(t)}\cdot \Gamma_0(t)=\big(\gamma(t),E(t),{\rm e}^{f(t)}\psi(t),\chi(t){\rm e}^{-f(t)}\big)
\end{equation*}
with $f(t)$ some differentiable complex-valued function satisfying $f(0)=0$. This $f$ depends solely on the ansatz $\psi$ by imposing the lift condition
\begin{equation*}
 0=(\omega_H)_{\Gamma(t)}\big(\dot{\Gamma}(t)\big)=\chi(t){\rm e}^{-f(t)}\frac{\dif}{\dif t}\big({\rm e}^{f(t)}\psi(t)\big)=\dot{f}(t)+\chi(t)\dot{\psi}(t).
\end{equation*}
Solving for $f(t)$, we thus find the actual state from the ansatz by application of the scale factor~${\rm e}^{f(t)}$;
\begin{equation} \label{eq:lift from geo correction}
 \Psi(t)=\exp\bigg({-}\int_0^t \chi(t')\dot{\psi}(t') \dif t'\bigg)\psi(t).
\end{equation}

This is the general expression for a state $\psi_0$ that undergoes parallel transport along a path~$\gamma$ or~$\bar{\gamma}$, also in the non-Hermitian case, expressed using an ansatz. Although the integral $\int_0^t \chi(t')\dot{\psi}(t')\dif t'$ is famous for its relation to the geometric phase, without additional assumptions it does not have physical significance. Indeed, the integral is not ansatz independent. That is, given another ansatz $\psi'(t)$, then $\psi'(t)=\exp(a(t))\psi(t)$ for some complex-valued function $a=a(t)$, and $\int_0^t \chi(t')\dot{\psi}(t')\dif t'$ becomes $\int_0^t \chi'(t')\dot{\psi}'(t')\dif t'+[a(t)-a(0)]$, which can in principle be any complex-valued continuous function. The actual state $\Psi(t)$ is of course ansatz independent; assuming the same initial condition, i.e., $\psi'(0)=\psi_0$, then $a(0)=0$ and as desired we obtain
\begin{gather*}
 \Psi'(t)=\exp\bigg({-}\!\int_0^t \chi'(t')\dot{\psi}'(t')\dif t'\bigg)\psi'(t)=\exp\bigg({-}\!\int_0^t \chi(t')\dot{\psi}(t')\dif t'-a(t)\bigg){\rm e}^{a(t)}\psi(t)=\Psi(t).
\end{gather*}
This confirms that the integral is, in general, only a correction factor to the specific ansatz $\psi(t)$; it merely quantifies how much our ansatz $\psi(t)$ was away from being the lift $\Psi(t)$.

In fact, just from physical arguments one should not expect the integral to yield an observable; we did not assume the state to be cyclic, so that there is no particular phase the integral can be equal to. Of course, if the state is cyclic, say the state returns at time $T$, then a geometric phase is well-defined. With an additional assumption on the ansatz $\psi$, namely $\psi(T)=\psi_0$, this geometric phase is indeed given by the integral at $t=T$. However, even in this cyclic case, if we consider intermediate times $t$, i.e., $0<t<T$, the integral does not yield a physical quantity; we fixed $a(T)=0$, but for intermediate times $a(t)$ can still be anything. In~other words, concerning a state acquiring a geometric phase, there is no well-defined rate at which this happens. This reflects that a geometric phase depends only on the locus of a (closed) path, not its parametrization. Still, the integral $\int_0^t \chi'(t')\dot{\psi}'(t')\dif t'$ can be related to path ``lengths'', as we consider in Section~\ref{sec:geo phase and distance}.

We will consider the state evolution in more detail in the following. We will distinguish between the cyclic and the non-cyclic case. The approach is summarized in diagram~\eqref{eq:diag Eig(H) Spec(H) N(H)} below, which is the pullback of diagram~\eqref{eq:Eig(V) Spec(V) N(V) triangle} along $H$. Namely, we see that the projection $\pi_{\lambda v}^H\from \Eig(H) \to N(H)$ can be considered as the main bundle to model the evolution of states. Indeed, given any loop in $N(H)$, it will certainly induce a holonomy operation on $\Eig(H)$, regardless of states being cyclic or non-cyclic. In contrast, one can also consider the bundle $\pi_v^H\from \Eig(H) \to \Spec(H)$. Clearly, loops in $\Spec(H)$ correspond to cyclic evolution only.\footnote{In an adiabatic setting, with cyclic we also assume that the system parameters are restored.} However, this bundle is principal and so has the advantage that its holonomies are easier to describe. We thus remark that for cyclic states one uses $\pi_v^H$, and one uses $\pi_{\lambda v}^H$ primarily for non-cyclic states. We also emphasize that the calculation of $\Psi(t)$ above can be used in both cases:
\begin{equation} \label{eq:diag Eig(H) Spec(H) N(H)}
 \begin{tikzcd}[column sep=1em,ampersand replacement=\&]
 \Eig(H) \ar{rr}{\text{cyclic states}} \ar[swap]{rd}{\text{non-cyclic states}} \&\& \Spec(H) \ar{ld}{\text{energy swaps}} \&\&\& (x,E,\psi)\ar[mapsto]{rr}{\pi_v^H} \ar[mapsto,swap]{rd}{\pi_{\lambda v}^H} \&\& (x,E) \ar[mapsto]{ld}{\pi_\lambda^H}\\
 \& N(H), \&\&\&\&\& x.
 \end{tikzcd}
\end{equation}

\subsubsection{The cyclic case}

Let us start with the cyclic case. As said, this concerns the bundle $\pi_v^H\from \Eig(H) \to \Spec(H)$. We remark that the only difference between a path in $N(H)$ and a path in $\Spec(H)$ is that the latter not only specifies the change in system parameters, but also which energy level is of interest. Moreover, as we have seen, fixing an initial energy $E_0$ uniquely specifies a path in $\Spec(H)$ given a path in $N(H)$. It is thus natural to consider a loop $\bar{\gamma}$ in $\Spec(H)$ as the input data for cyclic evolution.

Let us now consider the evolution of the eigenstates. Denote the basepoint of $\bar{\gamma}$ again by $(x_0,E_0)$, and assume $\bar{\gamma}$ returns to this point at time $t=T$. The state evolution is then described by the parallel transport map $P_{\bar{\gamma}}$, e.g., $P_{\bar{\gamma}}(\psi_0)=\Psi(T)$ following the previous calculation. As $\pi_v^H$ defines a principal bundle and $\omega_H$ is a principal connection, $P_{\bar{\gamma}}$ follows from standard holonomy theory. That is, $P_{\bar{\gamma}}$ is an automorphism of the $\C^\times$-torsor $\Eig_{E_0}(H(x_0))$, and thus amounts to scaling by a unique element in $\C^\times$. This element is clearly the geometric phase factor that any state in $\Eig_{E_0}(H(x_0))$ acquires due to following $\bar{\gamma}$. As we allow for non-Hermitian Hamiltonians, this phase factor need not be unitary. We hence obtain a definition of (generalized) geometric phase from holonomy as follows.
\begin{Definition}[geometric phase]
 Let $\bar{\gamma}$ be a loop in $\Spec(H)$ based at $(x_0,E_0)$. The geometric phase $\gamma_{\mathrm{geo}}$ due to $\bar{\gamma}$ is defined, modulo $2\pi$, via
 \begin{equation*}
 P_{\bar{\gamma}}= {\rm e}^{{\rm i}\gamma_{\mathrm{geo}}} \cdot \id_{\Eig_{E_0}(H(x_0))}.
 \end{equation*}
\end{Definition}

It now remains to calculate $\gamma_{\mathrm{geo}}$ explicitly. For this, we can use the earlier result from the ansatz technique. As usual, one only has to compare the phase difference between the final state~$\Psi(T)$ and the initial state $\psi_0$, i.e., $\Psi(T)=P_{\bar{\gamma}}(\psi_0)=\exp({\rm i}\gamma_{\mathrm{geo}}) \psi_0$. We thus consider the unique non-zero complex scalar $[\Psi(T)/\psi_0]$ such that $\Psi(T)=[\Psi(T)/\psi_0]\psi_0$. Equating $\exp({\rm i}\gamma_\mathrm{geo})=[\Psi(T)/\Psi(0)]$ and substituting our expressing for $\Psi(t)$ in terms of our ansatz $\psi(t)$, we find the expression
\begin{equation} \label{eq:geo phase general}
 \gamma_\mathrm{geo}={\rm i}\int_0^T \chi(t)\dot{\psi}(t)\dif t-{\rm i}\ln([\psi(T)/\psi_0]),
\end{equation}
where the logarithm term yields the usual modulo $2\pi$ of a phase. By construction this is invariant under both replacing $\psi_0$ with another state in the same ray and picking another ansatz $\psi$.

We see two terms of different nature in equation~\eqref{eq:geo phase general}. The integral term is clearly the usual integral for the geometric phase, and corrects for changes of $\psi$ along its own direction. The logarithm term is a correction for $\psi$ not closing on itself. Observe that only the two terms together are invariant under changing the choice of the ansatz $\psi$. The logarithm term vanishes whenever $\psi(1)=\psi(0)$, which happens, e.g., if $\psi$ is built using a local eigenstate. On the other hand, the integral term vanishes, e.g., when $\psi$ is chosen to be the lift $\Psi$, in which case the integrand is identically zero. In this case $\gamma_\mathrm{geo}$ can be computed using the end points only, see Example~\ref{ex:DP 1} below for an explicit example.

\begin{Example} \label{ex:DP 1}
 Let us consider a typical system with a diabolic point (DP), named after the diabolo shape of the energy bands \cite{Berry1984DiabolicTriangles}. Let us pick $V=\C^2$ with standard basis $(e_1,e_2)$, in which the Hamiltonian family reads
 \begin{equation*}
 H(a,b)=\begin{pmatrix} a& b\\b & -a \end{pmatrix}\!,
 \end{equation*}
 where $a$, $b$ are real numbers, i.e., $M=\R^2$. The energy bands are given as $\lambda_\pm(a,b)=\pm\sqrt{a^2+b^2}$, which have a single degeneracy at the origin. This is located at the apex of the diabolo, and is the DP of this system.

 Let us show how the geometric phase integral can be seen as a correction factor. Therefore, let us traverse the unit circle using the path $\gamma(t)=(\cos(t),\sin(t))$ with time interval $[0,2\pi]$, and consider the $\lambda_+$ level ($\lambda_-$ is similar). By looking at $H(\gamma(t))-\lambda_+(\gamma(t))I$, one readily finds an~nsatz $\psi(t)$ for the evolution of the eigenstate, together with a left-eigenstate path $\chi(t)$, as
 \begin{equation*}
 \psi(t)=\frac{1}{2}\begin{pmatrix}1+\cos(t) \\\sin(t) \end{pmatrix}\!, \qquad \chi(t)=\frac{1}{1+\cos(t)}\begin{pmatrix}1+\cos(t) \\\sin(t) \end{pmatrix}^\mathrm{T}\!.
 \end{equation*}
 The factor $1/2$ is introduced to have $\psi(0)=e_1$. Note that $\psi$ and $\chi$ are defined only for $t\in [0,\pi)$; being a (left-)eigenstate, they are not allowed to vanish. Still, for these times we can calculate the lift of $\gamma(t)$ starting at $e_1$.

 We can correct the ansatz $\psi(t)$ following equation~\eqref{eq:lift from geo correction}. We thus evaluate the geometric phase integral, whose integrand is
 \begin{equation*}
 \chi(t) \dot{\psi}(t)=\frac{-\sin(t)}{2(1+\cos(t))}=\frac{1}{2}\frac{\dif}{\dif t}\ln(1+\cos(t)).
 \end{equation*}
 Hence we find the lift of $\gamma$ starting at $e_1$ to be
 \begin{equation*}
 \Psi(t)=\sqrt{\frac{1}{2(1+\cos(t))}}\begin{pmatrix}1+\cos(t) \\ \sin(t) \end{pmatrix}=\begin{pmatrix}\cos(t/2) \\ \sin(t/2) \end{pmatrix}\!.
 \end{equation*}
 As the scale factor is real, we can interpret it as a length correction. This length interpretation clearly shows that our original $\psi(t)$ was varying along itself.

 Note that our final expression of $\Psi(t)$ can be extended to arbitrary times, hence yields the full lift of $\gamma$. Consequently, it is now convenient to obtain the geometric phase via $\Psi(t)$. For lifts, only the logarithm term in equation~\eqref{eq:geo phase general} contributes, and we find the phase (modulo $2\pi$)
 \begin{equation*}
 \gamma_\mathrm{geo}=-{\rm i}\ln([\Psi(2\pi)/\psi_0])=-{\rm i}\ln([-e_1/e_1])=-{\rm i}\ln(-1)=\pi.
 \end{equation*}
 In the perspective of loops in $\Spec(H)$, the above concerns the loop $\bar{\gamma}_+(t):=(\gamma(t),\lambda_+(\gamma(t)))=(\gamma(t),1)$, revealing that $P_{\bar{\gamma}_+}=-\id$. A similar argument shows that the loop $\bar{\gamma}_-(t)=(\gamma(t),-1)$, i.e., considering the other energy band, yields $P_{\bar{\gamma}_-}=-\id$ as well.
\end{Example}

In some situations, one can express the geometric phase alternatively as a curvature integral. A standard argument is as follows. Let $U\subset N(H)$ be a neighborhood on which a local eigenstate $\psi=\psi(x)$ is defined, which has as partner the local left-eigenstate $\chi$. If $\gamma$ is a loop in $U$ that forms the boundary of a surface $S$ also contained in $U$, then the geometric phase is given by
\begin{equation*}
 \gamma_{\mathrm{geo}}={\rm i}\int_\gamma \chi \dif \psi={\rm i}\int_{\partial S} \hat{\psi}^*\omega_H={\rm i}\int_S \hat{\psi}^*(\dif \omega_H)={\rm i}\int_{\hat{\psi}(S)} K_H.
\end{equation*}
We see that the assumption of $\gamma$ being a boundary is essential. If $\gamma$ is a loop but not a boundary, then the geometric phase is not given by a curvature integral. This motivates us to consider homology theory.

In order to do this, it is convenient to work on $\Spec(H)$ instead. Indeed, the integral above can be rewritten as an integral over $k_H$, which is the reduced curvature on $\Spec(H)$. Clearly, the local eigenstate $\psi(x)$ defines an energy function $E(x)$. The map $\hat{E}\from U\to \Spec(H)$, $x\mapsto (x,E(x))$ is then a local section of $\Spec(H)$. Moreover, as $\hat{E}=\pi_v^H\circ \hat{\psi}$, where $\pi_v^H$ is the projection $\Eig(H)\to \Spec(H)$, we have
\begin{equation*}
 \gamma_{\mathrm{geo}}={\rm i}\int_S \hat{\psi}^*(K_H)={\rm i}\int_S \hat{\psi}^*\big(\big(\pi_v^H\big)^*(k_H)\big)={\rm i}\int_S \hat{E}^*(k_H)={\rm i}\int_{\hat{E}(S)} k_H.
\end{equation*}

Let us interpret this result. First, we remark that the integral over $k_H$ is manifestly independent of the chosen local eigenstate $\psi$. More precisely, only the energy bands matter, not the exact eigenstates. In addition, we now consider surfaces in $\Spec(H)$, hence the above argument can be used whenever our path $\bar{\gamma}$ is a boundary. We thus find that the homology of $\Spec(H)$, i.e., the homology of the energy bands, plays a key role in the relation of geometric phase to curvature. Moreover, in case $\Eig(H)$ is flat, the homology theory allows one to see the topological nature of the geometric phase in an explicit way, as shown in the following.
\begin{Proposition}
 Let $\bar{\gamma}$ be a loop in $\Spec(H)$. If $\bar{\gamma}$ is the boundary of a surface $\Sigma$ in $\Spec(H)$, then the geometric phase acquired by traversing $\bar{\gamma}$ equals ${\rm i}\int_\Sigma k_H$.
\end{Proposition}

\begin{Corollary}
 If $\Eig(H)$ is flat, then the geometric phase due to a loop in $\Spec(H)$ only depends on the class of the loop in the first homology group $H_1(\Spec(H))$.
\end{Corollary}

\begin{Example} %\label{ex:DP 2}
 Let us continue Example~\ref{ex:DP 1}. As $H$ is symmetric, by Proposition~\ref{prop:flat connection} the connection is flat. Hence the geometric phases only depend on the homology of $\Spec(H)$, and thus are of topological nature. Clearly, $\Spec(H)$ is the diabolo minus the DP, which is homeomorphic to two punctured planes. Hence $H_1(\Spec(H))\cong \Z^2$, generated by (the classes of) the images of the unit circle under $\lambda_+$ and $\lambda_-$. It thus suffices to know the phase due to each generator, which is $\pi$ following our earlier calculation in Example~\ref{ex:DP 1}.
\end{Example}

It is now straightforward to recognize the bundles introduced by Simon \cite{Simon1983HolonomyPhase} in the bundle $\Eig(H)\to \Spec(H)$. According to Corollary~\ref{cor:real eigenvalues->trivial}, for Hermitian $H$ the bundle $\Spec(H)\to N(H)$ is trivial, i.e., $\Spec(H)$ separates into $n$ distinct energy bands. Let us label these energy bands by $1,\dots,n$ and write $\Spec_k(H)$ for energy band $k$, and similarly $\Eig_k(H)$ for the subbundle over $\Spec_k(H)$. Then $\Eig_k(H)\to\Spec_k(H)$ is a principal $\C^\times$-bundle, which is the bundle used by Simon for energy level $k$ adapted to our language. It is thus geometrically clear why this approach does not apply to non-Hermitian Hamiltonians; in that case the energy bands in $\Spec(H)$ need not be separated but connected via ``spiral staircase'' like structures. If this is the case, i.e., if non-cyclic states appear, then the bundle $\Eig(H)\to \Spec(H)$ can still model these states by parallel transport, but as the path $\bar{\gamma}$ is then not a loop this does not fit in the framework of holonomy.

\subsubsection{The non-cyclic case}

Let us now demonstrate how non-cyclic states do have a holonomy description when using the other projection, i.e., the bundle $\pi_{\lambda v}^H \from \Eig(H) \to N(H)$. Here, holonomy is to be understood in the context of a semi-principal bundle. Similarly to the case above, parallel transport along a~loop $\gamma$ in $N(H)$ based at $x_0$ induces the parallel transport map $P_\gamma$, which is an automorphism of the $\C^\times$-semi-torsor $\Eig(H(x_0))$. The set of all $P_\gamma$ for a fixed base point $x_0$ then defines the holonomy group
\begin{equation*}
 \Hol^{\Eig(H)}_{N(H)}(x_0)=\set{P_\gamma \in \Aut(\Eig(H(x_0)))}{\gamma\in \mathrm{Loop}(N(H),x_0)}.
\end{equation*}
However, these automorphisms are harder to describe. Whereas we could previously identify a~map $P_{\bar{\gamma}}$ with an element in $\C^\times$, this need not be for $P_\gamma$. In particular, by definition a non-cyclic state does not return to the same group orbit, i.e., the initial eigenray, hence no element in $\C^\times$ can relate the initial and final states. Instead, one must find how the eigenrays are transported individually, as illustrated in the following example.

\begin{Example}
 Let us consider a standard EP2 example, using the Hamiltonian family of Example~\ref{ex:start EP2}. We take $x_0=0$ again as our reference point, and from there we encircle the EP at $x=+{\rm i}$ in positive direction by following a path $\gamma$. One can take $\gamma$ to be a circular path, but as the connection is flat by Proposition~\ref{prop:flat connection}, the exact shape of $\gamma$ does not play a role.

 The evolution of the states due to encircling $\gamma$ is captured by the map $P_\gamma \from \Eig(H(0)) \to \Eig(H(0))$. Clearly, $\Eig(H(0))$ is the disjoint union of the two eigenrays, i.e., \begin{equation*}
 \Eig(H(0))=\C^\times e_1 \sqcup \C^\times e_2,
 \end{equation*}
 where $e_1=(1,0)^T$ and $e_2=(0,1)^T$ are the standard basis vectors. Hence, an element of $\Eig(H(0))$ is of the form $ze_1$ or $ze_2$ with $z\in \C^\times$. We emphasize that $\Eig(H(0))$ is a semi-torsor and not a vector space; linear combinations of $e_1$ and $e_2$ are not present, as these are not eigenstates of $H(0)$. Accordingly, $P_\gamma$ is equivariant and not linear.

 We can now calculate $P_\gamma$ as follows. First, it is sufficient to know $P_\gamma(e_1)$ and $P_\gamma(e_2)$ due to equivariance; $P_\gamma(ze_k)=zP_\gamma(e_k)$ for $k=1,2$. That is, it is sufficient to follow $e_1$ and $e_2$ around the EP. This can be done using an explicit parametrization of $\Spec(H)$. One can easily find eigenstates depending on $x$, and following Proposition~\ref{prop:flat connection} we ``normalize'' them so that no geometric phase will appear.\footnote{Note that this uses flatness of the connection. In the non-flat case the geometric phase depends on the exact shape of $\gamma$ and so cannot be expressed using local eigenstates. One should then use, e.g., the ansatz method instead.} Hence we arrive at the expressions
 \begin{equation*}
 \psi_\pm(x)= \mp \frac{1}{\sqrt{2(1+x^2\mp\sqrt{1+x^2}})}\begin{pmatrix}
 -x\\1\mp\sqrt{1+x^2}\end{pmatrix}\!,
 \end{equation*}
 where we remark that $\psi_+$ has a removable singularity at $x=0$; the limit reads $\lim_{x\to 0}\psi_+(x)\allowbreak=e_1$. Hence at $x_0=0$, $\psi_+(0)=e_1$ and $\psi_-(0)=e_2$ as desired. The lifts of these states to $\Eig(H)$ can thus be found by inspection of $\psi_\pm(x)$. Encircling the EP in positive direction swaps the $\pm$ signs, and in addition the overall root in the denominator obtains a factor of ${\rm i}$. Thus, after following $\psi_\pm(x)$ around the EP, we return with ${\rm i}^{-1}\psi_\mp(x)$. At our reference $x_0=0$, this rule becomes $e_1 \mapsto -{\rm i}e_2$, $e_2 \mapsto -{\rm i}e_1$. This information is enough to specify $P_\gamma$, which we thus find to be given by
 \begin{equation} \label{eq:Pgamma EP2}
 P_\gamma \from \quad ze_1 \mapsto -{\rm i}ze_2, \qquad ze_2 \mapsto -{\rm i}ze_1.
 \end{equation}
\end{Example}

In order to study $P_\gamma$, let us consider the following statement showing how a general automorphism of a semi-torsor can be studied by its invariant subspaces, similar to linear algebra.
\begin{Proposition} \label{prop:semi-torsor invariant subspaces}
 Let $G$ be a Lie group, $F$ a $G$-semi-torsor and $\phi \from F\to F$ an automorphism. Then $F=\sqcup_{i\in I'} F_i$, for some index set $I'$, decomposes $F$ into minimal $\phi$-invariant subspaces, i.e., $\phi(F_i)=F_i$ for all $i$. Moreover, if $G$ is Abelian and a particular $F_i$ consists of $k<\infty$ orbits, then $(\phi|_{F_i})^k$ equals translation by an element of $G$.
\end{Proposition}
\begin{proof}
 Clearly, the index set $I'$ is the original orbit space modulo the relation that orbits mapped into one another by $\phi$ are identified. If $F_i$ consists of $k$ orbits, then by minimality $(\phi|_{F_i})^k$ preserves the orbits in $F_i$. Picking $f\in F_i$, we find $\phi^k(f)=g(f) f$ for some $g(f)\in G$. If~$G$ is Abelian, then $g(f)$ is constant on the orbit through $f$. In addition, as equivariance yields $\phi^k(\phi(f))=g(f) \phi(f)$, we see $\phi(f)$ is scaled by the same element. It follows that $g(f)$ is constant on the subspace $F_i$.
\end{proof}

For the map $P_\gamma$, it is clear that a minimal subspace is any minimal union of eigenrays whose energies are permuted upon traversing $\gamma$. The element of $\C^\times$ associated to such a minimal union is a phase factor. Indeed, if there are $k$ rays in the union, then a ray first returns to itself by following $\gamma$ exactly $k$ times, after which is has obtained precisely this phase factor. Again, if $k>1$ the state is non-cyclic and there is no definite phase between an initial state $\psi_0$ in the union and its transport $P_\gamma(\psi_0)$. We come back to this in Section~\ref{sec:holonomy}, where we treat how $P_\gamma$ can be expressed via a holonomy matrix. Of course, if $k=1$, i.e., $\psi_0$ is cyclic, then we do recover the geometric phase. We summarize these findings in the following.

\begin{Corollary}
 Given a loop $\gamma$, the invariant subspaces of $P_\gamma$ are the minimal unions of eigenrays. If a union consists of $k$ eigenrays, then the characteristic phase of the union equals the geometric phase obtained by traversing the loop $\bar{\gamma}^k$, where $\bar{\gamma}$ is the lift of $\gamma$ to $\Spec(H)$ with an energy corresponding to any of the eigenrays in the union.
\end{Corollary}

It is also clear that $P_\gamma$ contains the information of the underlying permutation $p_\gamma$ of the energies. This is formally expressed by $\pi_v^H$ being a bundle map that is equivariant w.r.t.\ the collapsing quotient homomorphism $\C^\times \to 1$. As $\pi_v^H$ maps $\Eig(H)$ to $\Spec(H)$, this quotient reduces $P_\gamma$ to $p_\gamma$.
\begin{Proposition} \label{prop:Pgamma reduces to pgamma}
 For any $\gamma\in \Loop(N(H),x_0)$, the map $P_\gamma$ induces the following commutative diagram:
 \begin{equation*}
 \begin{tikzcd}[ampersand replacement=\&]
 \Eig(H(x_0)) \ar{r}{P_\gamma} \ar{d}{\pi_v^H} \& \Eig(H(x_0)) \ar{d}{\pi_v^H}\\
 \Spec(H(x_0)) \ar{r}{p_\gamma} \& \Spec(H(x_0)).
 \end{tikzcd}
 \end{equation*}
\end{Proposition}

\subsection{Including the dynamical phase} %\label{sec:dynamical phase}

The geometric properties of adiabatic dynamics for non-degenerate operators, Hermitian or non-Hermitian, can be described on $\Eig(V)$ using the connection $\omega$. This leaves out an important non-geometrical property, namely the dynamical phase. Nevertheless, $\Eig(V)$ does support a~calculation of the dynamical phase. This builds on a particular complex-valued function, which simply extracts the eigenvalue from the elements in $\Eig(V)$, i.e.,
 \begin{align*}
 \mathcal{E} \from\ \Eig(V) &\to \C,
 \\
 (A,\lambda,v) &\mapsto \lambda.
 \end{align*}
Given a lift $\Gamma$ in $\Eig(V)$, the corresponding dynamical phase is then the integral
\begin{equation*}
 \gamma_\mathrm{dyn}=-\frac{\rm i}{\hbar}\int \mathcal{E}(\Gamma(t)) \dif t.
\end{equation*}

It is even possible to put the dynamical phase explicitly in the lift. Clearly, the lift $\Gamma$ without dynamical phase vanishes under the covariant derivative
\begin{equation*}
 D={\rm d}+\omega.
\end{equation*}
If one modifies this to
\begin{equation*}
 D_\mathrm{tot}={\rm d}+\omega+\frac{\rm i}{\hbar}\mathcal{E}
\end{equation*}
and calculates the lift $\Gamma_\mathrm{tot}$ of $\gamma$ given by $D_\mathrm{tot}\Gamma_\mathrm{tot}=0$, then $\Gamma_\mathrm{tot}$ is the path of the eigenstate, including the dynamical phase, assuming $\gamma$ is parametrized by physical time.

\subsection{Relation with the work of Aharonov and Anandan} \label{sec:AA}

We found that $\Eig(V)$ provides a framework suitable for any non-degenerate finite-dimensional Hamiltonian, including non-Hermitian Hamiltonians in particular. However, this brings us to the question how the theory of Hermitian systems relates to $\Eig(V)$. The geometric framework for such Hermitian cases was pioneered by Aharonov and Anandan in \cite{Aharonov1987PhaseEvolution} (see also the elaboration in \cite{Anandan1990GeometryEvolution}). The geometric spaces found there differ significantly from the spaces obtained here. Nevertheless, we will show that they can be obtained from $\Eig(V)$.

Let us summarize the theory of \cite{Aharonov1987PhaseEvolution}, rephrasing it in line with our approach to $\Eig(V)$. First, the state space $V$ is now assumed to be a Hilbert space, i.e., $V$ should be equipped with a~Her\-mi\-tian inner product $\inn{\;}{\;}$. This allows one to define the unit sphere $S_1(V)$ inside $V$ by restricting to norm 1 states. As norm 1 fixes a state up to a~$U(1)$-phase, $S_1(V)$ is naturally a $U(1)$-manifold, and the quotient is the projective space $P(V)$ of $V$. This quotient defines the principal bundle
\begin{equation} \label{eq:AA bundle}
 \begin{tikzcd}
 U(1) \ar{r} & S_1(V) \ar{r} & P(V),
 \end{tikzcd}
\end{equation}
which is the central object in the formalism. The base $P(V)$ can be viewed as the space of rays in $V$, but also as the set of projectors projecting to a line. That is, the ray through $\ket{\psi}\in S_1(V)$ can be identified with the projection operator $\ket{\psi}\bra{\psi}$. The metric on $V$ obtained from the inner product $\inn{\;}{\;}$ induces a metric tensor on $S_1(V)$, and hence a connection 1-form. On a path $\ket{\psi(t)}$ in $S_1(V)$, this connection yields $\inn{\psi(t)}{\dot{\psi}(t)}$. If one lifts a loop from $P(V)$ to $S_1(V)$, it follows that the final state must lie in the same ray as the initial state. In this case, one says that the state is cyclic. If the path $\ket{\psi(t)}$ satisfies $\ket{\psi(1)}=\ket{\psi(0)}$, then the obtained Aharonov--Anan\-dan~(AA) phase is given by the integral ${\rm i}\int_0^1 \inn{\psi(t)}{\dot{\psi}(t)}\dif t$. This coincides with the adiabatic Berry phase if $\ket{\psi(t)}$ evolves adiabatically.

The AA phase is thus a generalization of the Berry phase from adiabatic state evolution to any path of states. Hence the AA phase is viewed as a non-adiabatic generalization. Still, we argue that it can be obtained from $\Eig(V)$, equipped with the ``adiabatic'' connection $\omega$. That is, we will show that the bundle in equation~\eqref{eq:AA bundle}, including connection, can be obtained by collapsing the principal bundle $\pi_v \from \Eig(V) \to \Spec(V)$. This can be done in two steps. First, we use the inner product to restrict the operators to the Hermitian ones and the vectors to normalized ones. After this, listing the operator will be redundant, and the second step is to discard it.

Let us describe the first step. The only additional ingredient we use is the chosen inner product $\inn{\;}{\;}$ on $V$. It allows us to talk about Hermitian operators, and we restrict $N(V)$ accordingly to the closed subset
\begin{equation*}
 N_{\inn{\;}{\;}}(V):=\set{A\in N(V)}{A^\dagger=A \text{ with respect to }\inn{\;}{\;}}
\end{equation*}
of all non-degenerate Hermitian operators on $V$. This is a non-canonical subset of $N(V)$; a~dif\-fe\-rent inner product may yield a different subset. Consequently, there are the subbundles over $N_{\inn{\;}{\;}}(V)$ given by
\begin{gather*}
 \Spec_{\inn{\;}{\;}}(V):=\set{(A,\lambda)\in \Spec(V)}{A^\dagger=A \text{ with respect to }\inn{\;}{\;}},
 \\
 \Eig_{\inn{\;}{\;}}(V):=\set{(A,\lambda,v)\in \Eig(V)}{A^\dagger=A,\norm{v}=1 \text{ with respect to }\inn{\;}{\;}}.
\end{gather*}
The $\C^\times$-action on $\Eig(V)$ reduces to $U(1)$-phase rotation on $\Eig_{\inn{\;}{\;}}(V)$, which has quotient space $\Spec_{\inn{\;}{\;}}(V)$.

The key observation now is that, after this restriction, we no longer need to know $A$ and~$\lambda$ to compute the eigencovector and eigenprojectors. Indeed, given the normalized vector $\ket{v}$, the covector $\theta$ is $\bra{v}$ and the eigenprojector is $\ket{v}\bra{v}$, regardless of the exact $A$ and~$\lambda$. Hence, we are motivated to discard $A$ and~$\lambda$.
To do so, it is convenient to describe $\Spec(V)$ using projectors. Clearly, any pair $(A,\lambda)\in \Spec(V)$ defines an eigenprojector $P_{(A,\lambda)}$ projecting on the eigenspace of $A$ corresponding to $\lambda$. Conversely, given an eigenprojector $P$, the eigenvalue can be retrieved from the identity $AP=PA=\lambda P$. Hence we may write $(A,\lambda)\in \Spec(V)$ equivalently as $(A,P_{(A,\lambda)})$. For $\Spec_{\inn{\;}{\;}}(V)$, we may even write an element as $(A,\ket{v}\bra{v})$, where~$\ket{v}$ is a~normalized eigenvector of $A$.

We can now perform the second step, i.e., the reduction. Reducing $\Spec_{\inn{\;}{\;}}(V)$ to the space~$P(V)$ is straightforward using the projector description where the element $(A,\ket{v}\bra{v})$ goes to~$\ket{v}\bra{v}$. Reducing $\Eig_{\inn{\;}{\;}}(V)$ to $S_1(V)$ is similar; we only keep the vector. Together, these maps define a morphism of bundles as follows:
\begin{equation*}
 \begin{tikzcd}
 \Eig_{\inn{\;}{\;}}(V) \ar{d}{\pi_v} \ar[dashed]{r} & S_1(V) \ar{d} && (A,\lambda,\ket{v}) \ar[mapsto]{r} \ar[mapsto]{d} & \ket{v} \ar[mapsto]{d}\\
 \Spec_{\inn{\;}{\;}}(V) \ar[dashed]{r} & P(V), && (A,\ket{v}\bra{v}) \ar[mapsto]{r} & \ket{v}\bra{v}.
 \end{tikzcd}
\end{equation*}

Our final claim is that the connection $\omega$ on $\Eig(V)$ also carries over to $S_1(V)$. This again follows the same two steps. First, clearly $\omega$ restricts to a $U(1)$-connection on $\Eig_{\inn{\;}{\;}}(V)$. Second, this restricted form admits push-forward to $S_1(V)$. Concerning explicit formulas, this push-forward is given by substituting $\theta=\bra{v}$, which yields exactly the connection used to define the~AA phase. The following then summarizes these findings.

\begin{Proposition}
 The restricted projection $\Eig_{\inn{\;}{\;}}(V) \to \Spec_{\inn{\;}{\;}}(V)$ has a canonical reduction to the bundle $S_1(V) \to P(V)$. Moreover, the restricted connection on $\Eig_{\inn{\;}{\;}}(V)$ admits push-forward to $S_1(V)$, yielding the standard $U(1)$-connection as used for the AA phase.
\end{Proposition}

We conclude that the above projection allows us to translate the general theory of $\Eig(V)$ to more specific results, enabled by a Hermitian inner product. We will use the phrase ``in the Hermitian case'' to indicate such a passage has happened. We already saw that notationally this amounts to replacing $\theta$ by $\bra{v}$, but the broad picture contains more. For example, $\Eig(V)$ is locally partitioned according to eigenvalues, while this is completely absent in $S_1(V)$.

\subsection{Quantum geometric tensor}

We will show that the covariant derivatives naturally define a tensor on $\Eig(V)$, which is a~straightforward generalization of the quantum geometric tensor. We also consider its redu\-ced version on $\Spec(V)$. Before we discuss the tensor itself, we first introduce convenient bases of tangent spaces of $\Eig(V)$ and $\Spec(V)$. Afterwards we comment on a relation between geometric phase and distance.

\subsubsection{Bases for tangent space of spectrum and eigenvector bundles}

Let us formulate (complex) bases for the tangent spaces $T_{(A,\lambda,v)}\Eig(V)$ and $T_{(A,\lambda)}\Spec(V)$. The idea is that we pass the $v$, $\lambda$ and $A$ components one-by-one. In this way, both tangent spaces can be described in a similar way.

Let us start with the eigenvector part. Clearly, the main difference between the two tangent spaces is that $T_{(A,\lambda,v)}\Eig(V)$ has a tangent along the eigenray $\Eig_\lambda(A)$. This direction is naturally spanned by the fundamental tangent vector $\partial_v$ originating from the scaling action. An~advantage of picking $\partial_v$ is that we may pick the remaining tangent vectors in $T_{(A,\lambda,v)}\Eig(V)$ to be horizontal. These tangent are then the horizontal lifts of unique tangents in $T_{(A,\lambda)}\Spec(V)$, which means we cover both tangent spaces simultaneously.

We continue with the eigenvalue part. Obviously, a change of eigenvalue must be accompanied by a change of operator. Hence, let us change the operator only by what is absolutely necessary. Writing $P=v\theta$ for the eigenprojector of $A$ corresponding to $\lambda$, a shift in $\lambda$ is then given by the tangent
\begin{equation*}
 \partial_\lambda=\dtzero (A+tP,\lambda+t,v,\theta).
\end{equation*}
We must then find $n-1$ other tangent vectors for the other eigenvalues of $A$. These can be obtained via paths of the form $(A+tP',\lambda,v,\theta)$, with $P'$ the corresponding eigenprojector. Note that $\lambda$ is constant as we vary another eigenvalue, but still consider tangents at $(A,\lambda,v)$. The combination of these $n$ tangent vectors defines the tuple $\partial_{\tilde{\lambda}}$, where we pick an ordering of $\Spec(A)$. Writing the coefficients of these vectors as $\Delta \tilde{\lambda}$, again cf the ordering, we obtain the linear combination $\Delta \tilde{\lambda} \partial_{\tilde{\lambda}}$. The corresponding tangents in $T_{(A,\lambda)}\Spec(V)$ are similar.

It thus remains to describe all changes in operator and eigenvector, where the eigenvector should not change along itself. As we should now avoid to change the spectrum of $A$, let us use similarity transformations. That is, we conjugate by the operator ${\rm e}^{{\rm i}tX}$, where ${\rm i}X\in \End(V)$ can be viewed as a generator; the extra $\rm i$ is for later convenience. We then obtain the infinitesimal conjugation action (extending equation~\eqref{eq:GL(V)-action N(V)}), which we view as the linear map $\End(V)\to T_{(A,\lambda,v)}\Eig(V)$ given by
\begin{equation*}
 {\rm i}X\mapsto \dtzero \big({\rm e}^{{\rm i}tX}A{\rm e}^{-{\rm i}tX},\lambda,{\rm e}^{{\rm i}tX}v,\theta {\rm e}^{-{\rm i}tX}\big).
\end{equation*}
As is, this parametrization of $T_{(A,\lambda,v)}\Eig(V)$ is not compatible with our earlier choices. For instance, if ${\rm i}X=P$ we retrieve $\partial_v$, and for other eigenprojectors the obtained tangent vanishes. Hence we require $X$ to be free of eigenprojectors of $A$, which means that the matrix of $X$ w.r.t.\ any eigenframe of $A$ has zero diagonal. In this case, we say $X$ is $A$-free. Note that the obtained tangent is horizontal if and only if $\theta X v=0$, which is automatically satisfied for $A$-free $X$.

In summary, we may express a general element $\mathfrak{v} \in T_{(A,\lambda,v,\theta)}\Eig(V)$ as
\begin{equation*}
 \mathfrak{v}={\rm i}X+\Delta \tilde{\lambda}\partial_{\tilde{\lambda}}+z\partial_v,
\end{equation*}
where for $T_{(A,\lambda)}\Spec(V)$ the first two terms suffice. This is a basis if we impose $X$ to be $A$-free, in which case the terms span a subspace of dimension $n^2-n$, $n$ and 1, respectively. The values of $\dif v$ and $\dif \theta$ then read
\begin{equation*} %\label{eq:dv and dtheta explicitly}
 \dif v(\mathfrak{v})=\dtzero {\rm e}^{{\rm i}tX}{\rm e}^{tz} v=({\rm i}X+z)v, \qquad
 \dif \theta(\mathfrak{v})=-\theta({\rm i}X+z).
\end{equation*}
We finally remark that it is possible to not impose $X$ to be $A$-free; the map to $T_{(A,\lambda,v)}\Eig(V)$ is then still surjective, but no longer injective. This can be convenient in practice; $X$ can play the role of a Schr{\"o}dinger Hamiltonian, which need not be $A$-free. We will keep this in mind when describing the tensors in the following.

\subsubsection{Generalized quantum geometric tensor}

On the space $\Eig(V)$ there is a canonical tensor which in the Hermitian case reduces to the quantum geometric tensor (QGT). This QGT was first reported in \cite{Provost1980RiemannianStates}, where it was found by looking at infinitesimal distance between states. Its anti-symmetric part was later recognized to essentially be the Berry curvature, demonstrating its relevance to the geometric phase, while its symmetric part yields a metric on parameter space \cite{Berry1989TheAfter}. We now show that $\Eig(V)$ supports a~more general tensor, which relates directly to covariant derivatives. Moreover, this generalized tensor makes no reference to an inner product. It is thus also incorporates a generalization of the QGT based on $\PT$-symmetry as reported in \cite{Zhang2019QuantumMechanics}.

Let us start from the standard expression. Fixing local coordinates on $N(H)$ and a local normalized eigenstate $\ket{\psi}=\ket{\psi(x)}$, the QGT is given by
\begin{equation*}
 T_{ij}=\bra{\partial_i\psi} (1-\ket{\psi}\bra{\psi})\ket{\partial_j\psi}\!.
\end{equation*}
We regard this to be the pull-back of a more abstract tensor $\G$ on $\Eig(V)$. Clearly, this $\G$ is simply given by
\begin{equation} \label{eq:QGT with projector}
 \G_{(A,v,\theta)}(\mathfrak{v}_1,\mathfrak{v}_2)=\dif\theta(\mathfrak{v}_1)(1-v\theta)\dif v(\mathfrak{v}_2), \qquad \mathfrak{v}_1, \mathfrak{v}_2 \in T_{(A,v,\theta)}\Eig(V).
\end{equation}
Because of this straightforward generalization from the Hermitian case to $\Eig(V)$, we refer to~$\G$ as the \emph{(generalized) quantum geometric tensor}. In addition, $\G$ is a natural tensor on $\Eig(V)$ in the following sense. Observe that the projector $1-v\theta$ is naturally obtained by taking the covariant derivative, as
\begin{equation*}
 \nabla v(\mathfrak{v})=\dif v(\mathfrak{v})-\omega(\mathfrak{v})v=\dif v(\mathfrak{v})-\theta(\dif v(\mathfrak{v}))v=(1-v\theta)\dif v(\mathfrak{v}), \qquad
 \mathfrak{v} \in T_{(A,v,\theta)}\Eig(V),
\end{equation*}
and similarly $\nabla \theta=\dif \theta (1-v\theta)$. We thus observe that $\G$ is the natural combination
\begin{equation*}
 \G_{(A,v,\theta)}(\mathfrak{v}_1,\mathfrak{v}_2)=\nabla \theta(\mathfrak{v}_1)(\nabla v(\mathfrak{v}_2)), \qquad
 \mathfrak{v}_1, \mathfrak{v}_2 \in T_{(A,v,\theta)}\Eig(V),
\end{equation*}
which directly displays its scale invariance.

We can write $\G$ more explicitly using our description of the tangent space $T_{(A,\lambda,v,\theta)}\Eig(V)$. Introducing labels 1 and 2 according to the two tangent vectors, this yields
\begin{equation*} %\label{eq:QGT using coefs}
 \G_{(A,v,\theta)}(\mathfrak{v}_1,\mathfrak{v}_2)=\theta(X_1X_2)v-(\theta X_1v)(\theta X_2v).
\end{equation*}
One can also write this using scale invariant quantities only. Writing points of $\Eig(V)$ as $(A,P,v,\theta)$, we find
\begin{equation*}
 \G_{(A,P,v,\theta)}(\mathfrak{v}_1,\mathfrak{v}_2)=\tr(PX_1X_2)-\tr(PX_1)\tr(PX_2).
\end{equation*}
This also provides an explicit form of the reduced QGT defined on $\Spec(V)$. We observe that in these expressions we need not impose the $X_i$ to be $A$-free. This is similar to correcting for a non-zero mean in probability theory; the second term provides a correction that vanishes for $A$-free $X_i$. Hence we choose not to impose the $A$-free condition, and instead keep the second term.

In the Hermitian case, the QGT is the sum of a symmetric tensor, known as the quantum metric, plus an anti-symmetric part proportional to the Berry curvature \cite{Berry1989TheAfter}. We find a similar decomposition here. The anti-symmetric part of $\G$ is readily seen to be $K/2$, as
\begin{equation*}
 \G^{\mathrm{alt}}_{(A,v,\theta)}(\mathfrak{v}_1,\mathfrak{v}_2)=\theta\frac{[X_1,X_2]}{2}v=\frac{1}{2}\dif \theta\wedge \dif v(\mathfrak{v}_1,\mathfrak{v}_2)=\frac{1}{2}K_{(A,v,\theta)}(\mathfrak{v}_1,\mathfrak{v}_2).
\end{equation*}
We may also write this as $\tr\big(P\frac{[X_1,X_2]}{2}\big)$, which in addition yields an explicit form of the reduced curvature $k$. The symmetric part $\mathcal{M}:=\G^\mathrm{sym}$ of $\G$ is then a tensor on $\Eig(V)$ generalizing the standard quantum metric tensor (QMT). Hence we will refer to $\mathcal{M}$ using the same name. A~scale invariant explicit form is
\begin{equation*} %\label{eq:dD def}
 \mathcal{M}_{(A,P,v,\theta)}(\mathfrak{v}_1,\mathfrak{v}_2)=\cov_P(X_1,X_2):=
 \tr\bigg(P\frac{\{X_1,X_2\}}{2}\bigg)-\tr(PX_1)\tr(PX_2),
\end{equation*}
where $\cov_P$ is the covariance of non-commutative operators w.r.t.\ the density $P$. We hence obtain a decomposition similar to \cite[equation~(30)]{Berry1989TheAfter}, but now for the generalized tensors on~$\Eig(V)$.
\begin{Proposition}
 The QGT on $\Eig(V)$ decomposes as a linear combination of the QMT and the curvature as
 \begin{equation*}
 \G=\mathcal{M} +\frac{1}{2}K,
 \end{equation*}
 which is the decomposition of $\G$ into its symmetric and anti-symmetric part, respectively.
\end{Proposition}

{\sloppy
Notable properties that do not generalize from the Hermitian case are the following. Clearly~$\mathcal{M}$ and $K$ are complex rather than real-valued tensors, and no longer the real resp.\ imaginary part of $\G$. In addition, $\mathcal{M}$ is a degenerate form, hence does not follow a standard metric interpretation.
We also wish to comment on reducing the QMT to a ``metric'' on parameter space. In~the Hermitian case this can be done for each energy band separately. That is, for a fixed energy band, one can define a global eigenstate and so obtain a metric on $N(H)$ by pull-back of $\M$. In~the non-Hermitian case, or better, whenever $\Spec(H)$ is non-trivial, this need not be possible as a global eigenstate could be unavailable.

}

\subsubsection{Relation between geometric phase and distance} \label{sec:geo phase and distance}

Let us comment on a relation between geometric phase and distance. This was pointed out in~\cite{Pati1991RelationEvolution} for the Hermitian case, but was also considered for the non-Hermitian case in \cite{Cui2014UnificationPhases}. The idea is to compare two distance functions defined on normalized states. Using a given inner product on $V$ and the corresponding bundle $S_1(V)\to P(V)$ from Section~\ref{sec:AA}, these distances are given~by
\begin{equation*}
 L(\psi,\psi')=\norm{\psi-\psi'}, \qquad
 D(\psi,\psi')=\sqrt{2-2\abs{\inn{\psi}{\psi'}}},\qquad
 \psi,\psi'\in S_1(V).
\end{equation*}
Here, $L$ is the inner product metric on $V$ restricted to $S_1(V)$. The ray space $P(V)$ also inherits a metric, whose pull-back to $S_1(V)$ is $D$. In this way, $L$ measures total length, whereas $D$ only measures the underlying change of rays. Their difference is then related to movement along the ray, which is where we find geometric phase. To make this concrete, one works infinitesimally. Setting $\psi'=\psi+\dot{\psi}\dif t+\tfrac{1}{2}\ddot{\psi}\dif t^2+\cdots$ and keeping terms up to $\dif t^2$ only, one finds $\dif L^2=\inn{\dot{\psi}}{\dot{\psi}}\dif t^2$ and $\dif D^2=\big[\inn{\dot{\psi}}{\dot{\psi}}+\inn{\psi}{\dot{\psi}}^2\big]\dif t^2$. One thus finds $\dif L^2-\dif D^2=({\rm i}\inn{\psi}{\dot{\psi}}\dif t)^2$, so that the geometric phase integrand equals $\sqrt{\dif L^2-\dif D^2}$, which is one of the main results in \cite{Pati1991RelationEvolution}. It is tempting to write ``$\dif \gamma_\mathrm{geo}$" for ${\rm i}\inn{\psi}{\dot{\psi}}\dif t$, but we refrain from doing that. Following the discussion after equation~\eqref{eq:lift from geo correction}, there is no definite rate at which a geometric phase is acquired, although the contrary is suggested by such notation. There is also no need to introduce more notation; the tensor field corresponding to this quantity is simply ${\rm i}\omega$.

We now generalize to the non-Hermitian setting of $\Eig(V)$. First, the distances $L$ and $D$ do not carry over. Concerning $L$, we see $L^2=2-\inn{\psi}{\psi'}-\inn{\psi'}{\psi}$ generalizes to the function $2-\theta_1(v_2)-\theta_2(v_1)$ on $\Eig(V)$, but its square root would be multi-valued on $\Eig(V)$, meaning that~$L$ itself does not generalize. Similar issues appear for $D^2=2-2\sqrt{\inn{\psi}{\psi'}\inn{\psi'}{\psi}}$. Nevertheless, the tensor fields $\dif L^2$ and $\dif D^2$ do have a generalization to $\Eig(V)$. One can check that $\dif D^2$ is actually $\M$ for the Hermitian case, see also \cite{Cui2012GeometricMechanics,Cui2014UnificationPhases}. Finally, $\dif L^2$ generalizes to the tensor field $\mathcal{L}:=(\dif \theta \dif v)^\mathrm{sym}$, which is the non-covariant counterpart of $\M$. The relation between $\mathcal{L}$, $\M$ and~$\omega$ is readily found from equation~\eqref{eq:QGT with projector}; first rewriting this to
\begin{equation*}
 \G=\dif \theta \dif v+(-\dif \theta v)\otimes(\theta \dif v)=\dif \theta \dif v+\omega\otimes \omega
\end{equation*}
and then taking the symmetric part yields
\begin{equation*}
 \M=\mathcal{L}+\omega\otimes\omega.
\end{equation*}
This is thus an equality of tensor fields on $\Eig(V)$, generalizing the relation ``$\dif D^2=\dif L^2-\dif \gamma_\mathrm{geo}^2$".

\section{Explicit description of the holonomy} \label{sec:holonomy}

In Sections~\ref{sec:eigval bundle and EPs} and \ref{sec:eigvec bundle and GPs} we saw how the bundles $\Spec(V)$ and $\Eig(V)$ provide a holonomy interpretation for the physics of geometric phases and exceptional points. However, in the non-cyclic case, we did not provide an easier description of the maps $p_\gamma$ and $P_\gamma$ in terms of explicit permutations like $(12)$ resp.\ holonomy matrices. We will now treat how these can be obtained, and afterwards show that this naturally follows holonomy theory as well.

\subsection{Explicit permutations and holonomy matrices} \label{sec:explicit perms and hol mats}

Let us start with the permutations of energies. Our goal is to find a formalism so that a~permutation $p_\gamma$ of the eigenvalues of some $A\in N(V)$ can be expressed using an standard permutation $\sigma\in S_n$. A first attempt would be to label the eigenvalues, i.e., call them $\lambda_1,\dots,\lambda_n$, and define~$\sigma$ by the relation $p_\gamma(\lambda_i)=\lambda_{\sigma(i)}$. However, this approach has a serious disadvantage. Although this works fine if we consider the eigenvalues of only the operator $A$, as the authors did in \cite{Pap2018Non-AbelianPoints}, this method does not extend to all of $N(V)$. The reason is that knowing the action of $S_n$ on a~spectrum is equivalent to having a labelling of the eigenvalues \cite{Pap2020FramesTheory}. For example, if there are 3 eigenvalues, then from the action of $(12)$ one finds which eigenvalue is labeled 3. Hence, if the above extends to a global $S_n$-action, then $\Spec(H)$ admits a global labelling and is thus trivial.

We thus wish to obtain a method that also applies to non-trivial $\Spec(H)$. Instead of considering eigenvalues separately, let us consider specific tuples of eigenvalues. Namely, for an operator $A$, we consider the $n$-tuples that list all eigenvalues exactly once, which brings us to the set
\begin{equation*}
 \Spec(A)!=\set{\tilde{\lambda}=(\lambda_1,\dots,\lambda_n)\in \Spec(A)^n}{\text{all eigenvalues of $A$ appear once in $\tilde{\lambda}$}}.
\end{equation*}
Observe that this set has $n!$ elements, hence the factorial notation. In addition, this space is naturally endowed with the $S_n$-action on $n$-tuples, cf.\ our earlier notation in equation~\eqref{eq:perm action Fr(V)timesCun}. It~is this $S_n$-action that will help us with describing the map $p_\gamma$. Of course, one can act with $p_\gamma$ on a tuple $\tilde{\lambda}$ entry-wise, which leads to the map $\tilde{p}_\gamma \from (\lambda_1,\dots,\lambda_n) \mapsto (p_\gamma(\lambda_1),\dots,p_\gamma(\lambda_n))$. As $\tilde{p}_\gamma\big(\tilde{\lambda}\big)$ is a reordering of $\tilde{\lambda}$, we obtain a unique permutation $\sigma$, depending on $\tilde{\lambda}$, by stating that
\begin{equation*}
 \tilde{p}_\gamma\big(\tilde{\lambda}\big)=\sigma \cdot \tilde{\lambda}.
\end{equation*}

Let us say a few more words on how the tuple method is subtly yet significantly different from the labelling method we started with. First, although we write the eigenvalues with indices (just to distinguish them), in principle they are never labelled. Indeed, like in Example~\ref{ex:EP2 yields (12)} below, we are interested in the \emph{values} of the energies, not the names we use for them. Of course, any tuple induces a labelling by assigning label $k$ to the energy appearing in position $k$, like we also used above. However, $\sigma$~is not based on these labels, but instead on the tuple positions. That is, $\sigma$~will send the energy at place $k$ to place $\sigma(k)$, regardless of what energy is located at place~$k$ or which value we designated as~$\lambda_k$.

\begin{Example} \label{ex:EP2 yields (12)}
 Let us continue with the EP2 example of Example~\ref{ex:start EP2}, where $E_+(x_0)$ and $E_-(x_0)$ are exchanged when one encircles an EP. Picking the ordering of $\Spec(H(x_0))$ to be $(E_+(x_0),E_-(x_0))$, the above reads
 \begin{equation*}
 \tilde{p}_\gamma(E_+(x_0),E_-(x_0))=(E_-(x_0),E_+(x_0))=(12)\cdot (E_+(x_0),E_-(x_0))
 \end{equation*}
 so that $\sigma=(12)$. In this case, as $n=2$, one also obtains $(12)$ for the alternative ordering $(E_-(x_0),E_+(x_0))$.

 Consider now the case where we add a separate constant energy level $E_0(x)=0$, i.e., consider the Hamiltonian $H'(x)=H(x)\oplus 0$ which brings us to $n=3$. Picking the ordering $(E_+(x_0),E_0(x_0),E_-(x_0))$ of $\Spec(H'(x_0))$ one finds
 \begin{gather*}
 \tilde{p}_\gamma(E_+(x_0),E_0(x_0),E_-(x_0))\!=(E_-(x_0),E_0(x_0),E_+(x_0))
\!=(13)\cdot (E_+(x_0),E_0(x_0),E_-(x_0)),
 \end{gather*}
 which thus yields $(13)$. Of course, in the ordering $(E_+(x_0),E_-(x_0),E_0(x_0))$ we would again obtain $(12)$.

 Finally, let us demonstrate how the tuple method differs from labelling the energies. Consider we label $E_+(x_0)=E_1$, $E_0(x_0)=E_2$ and $E_-(x_0)=E_3$. In the ordering $(E_1,E_2,E_3)$, clearly $\sigma$ will replace the labels as expected (up to a convention; here the label $E_k$ will become $E_{\sigma^{-1}(k)}$ instead of $E_{\sigma(k)}$). However, in another ordering this need not be so, for example
 \begin{equation*}
 (12)\cdot (E_3,E_1,E_2) = (E_1,E_3,E_2)
 \end{equation*}
 instead of the tuple $(E_3,E_2,E_1)$ one would obtain by checking the labels. The method of using a reference tuple thus circumvents the need for explicit labels.
\end{Example}

We can extend this approach to describe the evolution of states, as given by a map $P_\gamma$, by a~holonomy matrix. Similarly to the previous approach with eigenvalues, we will consider a~tuple formulation. In this case, for an operator $A\in N(V)$, we consider tuples $\tilde{f}=(f_1,\dots,f_n)$ of eigenvectors of $A$ such that each eigenray is represented exactly once. That is, $\tilde{f}$ should be an eigenframe of $A$, and we write $\EigFr(A)$ for the space of all eigenframes of $A$. This space has a~natural action by the wreath product $\struc$, isomorphic to the group of complex generalized permutation matrices, which we already encountered in Section~\ref{sec:param of non-deg}. In this case, the action reads\looseness=-1
\begin{align*}
 \struc \times \EigFr(A) &\to \EigFr(A),
 \\
 ((z_1,\dots,z_n),\sigma)\cdot (f_1,\dots,f_n) &= (z_1f_{\sigma^{-1}(1)},\dots,z_nf_{\sigma^{-1}(n)}),
 \end{align*}
i.e., the action is similar to the earlier $\struc$-action in equation~\eqref{eq:struc action Fr(V) times Cun}. Using this action, we can express a map $P_\gamma$ by an element $(\tilde{z},\sigma)\in \struc$ once we have chosen a reference eigenframe $\tilde{f}$. Writing $\tilde{P}_\gamma$ for the map on $\EigFr(A)$ obtained by applying $P_\gamma$ entry-wise, we obtain the group element by stating
\begin{equation*}
 \tilde{P}_\gamma(\tilde{f})= (\tilde{z},\sigma) \cdot \tilde{f}.
\end{equation*}
If one represents $(\tilde{z},\sigma)$ by a permutation matrix, which we will denote by $M_{\tilde{f}}$, then $M_{\tilde{f}}$ is the matrix of the linear extension of $P_\gamma$ to the entire state space. The dependence of $M_{\tilde{f}}$ on $\tilde{f}$ (again keeping the loop $\gamma$ fixed) follows the familiar conjugation rule.

\begin{Example}
 In case of the EP2, we found $P_\gamma$ as given by equation~\eqref{eq:Pgamma EP2}. With respect to the eigenframe $(e_1,e_2)$ of $H(0)$, we find
 \begin{equation*}
 \tilde{P}_\gamma(e_1,e_2)=(P_\gamma(e_1),P_\gamma(e_2))=(-{\rm i}e_2,-{\rm i}e_1)=((-{\rm i},-{\rm i}),(12))\cdot (e_1,e_2),
 \end{equation*}
 i.e., we find $((-{\rm i},-{\rm i}),(12)) \in \C^\times \wr I_2$. The holonomy matrix $M_{(e_1,e_2)}$ is thus
 \begin{equation*}
 M_{(e_1,e_2)}=
 \begin{pmatrix}
 -{\rm i} & 0\\
 0 & -{\rm i}\\
 \end{pmatrix}
 \begin{pmatrix}
 0 & 1\\
 1 & 0\\
 \end{pmatrix}
 =
 \begin{pmatrix}
 0 & -{\rm i}\\
 -{\rm i} & 0
 \end{pmatrix}\!.
 \end{equation*}
 Or, alternatively,
 \begin{equation*}
 \begin{pmatrix}
 P_\gamma(e_1)\\
 P_\gamma(e_2)\\
 \end{pmatrix}
 =
 \begin{pmatrix}
 -{\rm i}e_2\\
 -{\rm i}e_1\\
 \end{pmatrix}
 =
 \begin{pmatrix}
 0 & -{\rm i}\\
 -{\rm i} & 0
 \end{pmatrix}
 \begin{pmatrix}
 e_1\\
 e_2\\
 \end{pmatrix}\! .
 \end{equation*}

 The more familiar pattern $\psi_1\mapsto \psi_2 \mapsto -\psi_1$ is obtained by picking, e.g., $({\rm i}e_1,e_2)$ as the reference frame. Indeed, as
 \begin{equation*}
 \tilde{P}_\gamma({\rm i}e_1,e_2)=(e_2,-{\rm i}e_1)=((1,-1),(12))\cdot ({\rm i}e_1,e_2)
 \end{equation*}
 one finds the group element $((1,-1),(12))$ and the holonomy matrix reads
 \begin{equation*}
 M_{({\rm i}e_1,e_2)}=
 \begin{pmatrix}
 1 & 0\\
 0 & -1\\
 \end{pmatrix}
 \begin{pmatrix}
 0 & 1\\
 1 & 0\\
 \end{pmatrix}
 =
 \begin{pmatrix}
 0 & 1\\
 -1 & 0\\
 \end{pmatrix}\!.
 \end{equation*}
 The matrices are different only because we picked a different reference frame; both matrices still describe the map $P_\gamma$.
\end{Example}

\subsection{Holonomy interpretation on the frame bundle}

We now discuss how the approach of Section~\ref{sec:explicit perms and hol mats} complies with a more general theory on semi-torsors. As shown in \cite{Pap2020FramesTheory}, one can define the notion of a frame of a semi-torsor. Let us briefly review this here for completeness. If $G$ is a Lie group and $F$ a $G$-semi-torsor consisting of $n$ orbits, then $G\times I_n \cong F$ as (left) $G$-spaces. Such an isomorphism is always of the form $(g,i)\mapsto gf_i$ for a unique tuple $\tilde{f}=(f_1,\dots,f_n)$ of elements in $F$. This tuple $\tilde{f}$ we call a basis of $F$, similar to linear algebra. Observe that $n$ elements of $F$ form a basis if and only if each orbit in $F$ is represented exactly once. The set of all frames of $F$ we denote by $\Fr(F)$. Clearly, one can translate or permute the elements of a basis. This is captured by a natural action of the wreath product $G\wr I_n$, i.e.,
{\samepage\begin{align*}
 G\wr I_n \times \Fr(F) &\to \Fr(F),
 \\
 ((g_1,\dots,g_n),\sigma)\cdot (f_1,\dots,f_n) &= \big(g_1f_{\sigma^{-1}(1)},\dots,g_nf_{\sigma^{-1}(n)}\big).
 \end{align*}
In fact, this action is free and transitive, so that, in other words, $\Fr(F)$ is a $G\wr I_n$-torsor.

}

This is exactly what we did in the previous section. For the eigenvalues, we see that $\Spec(A)$ is a set containing $n$ points, hence trivially a semi-torsor with $G$ the trivial group. Each eigenvalue is thus an orbit by itself, and a frame is thus any tuple $\tilde{\lambda}=(\lambda_1,\dots,\lambda_n)$ listing all elements of $\Spec(A)$ once. These are exactly the tuples we considered before, i.e., \begin{equation*}
 \Fr(\Spec(A))=\Spec(A)!.
\end{equation*}
This space is naturally endowed with the $S_n$-action on $n$-tuples, making it an $S_n$-torsor. For the eigenstates, we consider the space $\Eig(A)$, which we already found to be a $\C^\times$-semi-torsor. The orbits in this case are the eigenrays, and hence a frame of $\Eig(A)$ is a tuple $\tilde{f}$ of eigenvectors such that each eigenray is represented exactly once. Indeed, we again find that $\tilde{f}$ should be an eigenframe of $A$, i.e.,
\begin{equation*}
 \Fr(\Eig(A))=\EigFr(A).
\end{equation*}
And this space is clearly a $\struc$-torsor.

The next step is to lift this procedure of taking frames to the level of semi-principal bundles. We remark that this theory closely resembles that of the frame bundle of a vector bundle. As~shown in \cite{Pap2020FramesTheory}, the operation of taking the frame space induces a functor from semi-principal bundles to principal bundles. Given a semi-principal $G$-bundle $\pi \from B\to M$, we obtain its frame bundle $\pi! \from \Fr(B) \to M$ by applying the above fiber-wise;
\begin{equation*}
 \Fr(B)=\bigsqcup_{m\in M} \Fr(B_m).
\end{equation*}
It holds that $\pi! \from \Fr(B)\to M$ is a principal $G\wr I_n$-bundle. Moreover, any $G$-connection on $B$ becomes a $G\wr I_n$-connection on $\Fr(B)$ by stating that a path in $\Fr(B)$ is horizontal if and only if all frame elements traverse horizontal paths. The map $B \mapsto \Fr(B)$ induces a retracting functor from the semi-principal bundles to the principal bundles. The holonomy interpretation of the explicit permutations and holonomy matrices we found above can now be obtained by applying this frame bundle functor to $\Spec(V)$ resp.\ $\Eig(V)$.

\subsubsection{Frame bundle of the spectrum bundle}

Let us start with $\Spec(V)$. We found in Theorem~\ref{thm:Spec(V) bundle over N(V)} that $\pi_v \from \Spec(V) \to N(V)$ is an $I_n$-bundle, hence it is a semi-principal bundle with trivial structure group. The frame bundle functor then yields the following.
\begin{Proposition}
 The frame bundle of $\Spec(V)$ is the principal $S_n$-bundle
 \begin{equation*}
 \En(V):=\Fr(\Spec(V))=\set{\big(A,\tilde{\lambda}\big)\in N(V)\times \C^n}{\tilde{\lambda}\in \Fr(\Spec(A))},
 \end{equation*}
 with action
 \begin{equation*}
 \sigma \cdot (\lambda_1,\dots,\lambda_n) = \big(\lambda_{\sigma^{-1}(1)},\dots,\lambda_{\sigma^{-1}(n)}\big)
 \end{equation*}
 and projection
 \begin{align*}
 \pi_\lambda! \from\ \En(V) &\to N(V),
 \\ \big(A,\tilde{\lambda}\big) &\mapsto A.
 \end{align*}
\end{Proposition}

Remarkably, we observe that $\pi_\lambda!$ is also the pull-back of $q\from \Cun \to \binom{\C}{n}$ along the map $\Spec \from N(V)\to \binom{\C}{n}$; the fiber $\Fr(\Spec(A))$ of $\En(V)$ equals $q^{-1}(\Spec(A))$. In addition, when a Hamiltonian operator family $H\from M \to \End(V)$ is given, it is again natural to consider the pull-back along $H$. This yields the bundle $\pi_\lambda^H! \from \En(H) \to N(H)$, $(x,\tilde{E}) \mapsto x$, where $\tilde{E}$ is then a~tuple listing the energies of $H(x)$, and is the frame bundle of $\pi_\lambda^H$. Observe that both pull-backs together result in the following commutative diagram, where all vertical maps are principal $S_n$-bundles:
\begin{equation*} %\label{diag:E(V) from q}
 \begin{tikzcd}
 \En(H) \ar{r}\ar{d}{\pi_\lambda^H!} & \En(V) \ar{r}\ar{d}{\pi_\lambda!} & \Cun \ar{d}{q} & (x,\tilde{E}) \ar[mapsto]{r}\ar[mapsto]{d} & (H(x),\tilde{E}) \ar[mapsto]{r}\ar[mapsto]{d} & \tilde{E}\ar[mapsto]{d}\\
 N(H) \ar{r}{H} & N(V) \ar{r}{\Spec} & \binom{\C}{n},& x \ar[mapsto]{r} & H(x) \ar[mapsto]{r} & \Spec(H(x)).
 \end{tikzcd}
\end{equation*}
We observe that $\En(H)$ and $\En(V)$ are also the natural spaces behind the merging path method for EP detection, which we reviewed in Section~\ref{sec:merging path method}. Namely, the above diagram can alternatively be obtained by completing diagram~\eqref{eq:merging path diagram} using pull-backs.

The holonomy interpretation behind the permutations can be explained as follows. Fix a~loop $\gamma$ in $N(H)$ based at $x_0$, and let $\tilde{E}$ be the ordering of $\Spec(H(x_0))$ w.r.t.\ which we find the permutation $\sigma$ expressing $p_\gamma$. Observe that the picking of the ordering $\tilde{E}=(E_1,\dots,E_n)$ of $\Spec(H(x_0))$ is equivalent to picking the point $\big(x_0,\tilde{E}\big)$ in the fiber of $\En(H)$ above $x_0$. This point $\big(x_0,\tilde{E}\big)$ fixes a unique lift of $\gamma$ to $\En(H)$. Physically, this lift traces all of the energies simultaneously and in a particular order. The permutation $\sigma$ now naturally appears via the holonomy theory of principal bundles; the end point of the lift is $\sigma\big(x_0,\tilde{E}\big)$, which fixes $\sigma$ uniquely. One can alternatively argue by means of holonomy groups. We already found that $p_\gamma$ lies in the holonomy group $\Hol^{\Spec(H)}_{N(H)}(x_0)$. Clearly, $p_\gamma$ extends to a map $\tilde{p}_\gamma$ on frames by applying~$p_\gamma$ on each frame element separately, and so $\tilde{p}_\gamma$ is an element of $\Hol^{\En(H)}_{N(H)}(x_0)$. Both of these groups are based on a point in the base space, and so consist of ``abstract'' fiber automorphisms. An~explicit group element is obtained by considering the holonomy group at a point in the total space. In~this case, we move to the holonomy group at $\big(x_0,\tilde{E}\big)$ defined as
\begin{equation*}
 \Hol^{\En(H)}\big(x_0,\tilde{E}\big)=\set{\sigma\in S_n}{\big(x_0,\tilde{E}\big)\sim \sigma\big(x_0,\tilde{E}\big)},
\end{equation*}
where the equivalence relation $\sim$ holds if and only if the two points can be connected by a path. The groups are related by the isomorphism $\Hol^{\En(H)}_{N(H)}(x_0) \to \Hol^{\En(H)}\big(x_0,\tilde{E}\big)$, $\tilde{p}_\gamma \mapsto \big[\tilde{p}_\gamma\big(\tilde{E}\big)/\tilde{E}\big]$, where $[-/-]$ denotes the unique group element translating the right entry to the left one. We~summarize these findings in the following.

\begin{Lemma}
 Given a loop $\gamma$ in $N(H)$ based at $x_0$, the permutation $\sigma$ expressing $p_\gamma$ w.r.t.\ the ordering $\tilde{E}$ of $\Spec(H(x_0))$ is the holonomy element at the point $\big(x_0,\tilde{E}\big)$ in the frame bundle~$\En(H)$.
\end{Lemma}

It follows that we should consider $\sigma$ to lie in the group $\Hol^{\En(H)}\big(x_0,\tilde{E}\big)$. This automatically accounts for the change in $\sigma$ if we choose another ordering $\tilde{E}'$. Geometrically, we see this amounts to picking a different reference point in $\En(H)$, namely $(x_0,\tilde{E}')$ instead of $\big(x_0,\tilde{E}\big)$. Hence, the permutation $\sigma'$ obtained w.r.t.\ $\tilde{E}'$ lies in $\Hol^{\En(H)}\big(x_0,\tilde{E}'\big)$ rather than $\Hol^{\En(H)}\big(x_0,\tilde{E}\big)$. We observe that the relation between $\sigma$ and $\sigma'$ is given by the usual conjugation relation between holonomy groups. Namely, there is a unique $\tau\in S_n$ such that $\tilde{E}'=\tau \tilde{E}$, and consequently $\sigma'=\tau \sigma \tau^{-1}$. Note that the frame dependence of $\sigma$ can also be regarded as a gauge dependence. The gauge freedom is then in the choice of frame of $\Spec(A)$, and different gauges are related by an $S_n$-symmetry. Of course, the gauge invariant behind the permutations $\sigma$ and $\sigma'$ is the map of eigenvalues $p_\gamma$, or equivalently its frame version $\tilde{p}_\gamma$.

Using the framework of holonomy groups, we can be more explicit on how permutations arising from different loops $\gamma$ are related, as the authors also studied in \cite{Pap2018Non-AbelianPoints}. Like $\Spec(H)$, also $\En(H)$ is a covering space of $N(H)$, and so the permutations are prescribed by the monodromy action. With respect to the ordering $\tilde{E}$ of $\Spec(H(x_0))$, this is captured by the homomorphism
\begin{align*}
 \pi_1(N(H),x_0) &\to \Hol^{\En(H)}\big(x_0,\tilde{E}\big),
 \\
 [\gamma] &\mapsto \big[\tilde{p}_\gamma\big(\tilde{E}\big)/ \tilde{E}\big].
 \end{align*}
Clearly this homomorphism is surjective, and its kernel consists of all classes whose paths induce a cyclic change of eigenvectors. In addition, it shows how the holonomy matrix of a complicated loop can be decomposed in holonomy matrices of more elementary loops. This present argument hence improves on a previous proof of this fact given in \cite{Pap2018Non-AbelianPoints}.

Let us briefly compare the holonomy groups of $\En(V)$ and $\En(H)$. For $\En(V)$, every holonomy group exhausts all of $S_n$ as all possible paths are present. Indeed, the bundle $\pi_\lambda!\from \En(V)\to N(V)$ does not admit a reduction of the structure group. However, $\En(H)$ can have smaller holonomy groups, and admit a reduction of the structure group. In an extreme case, e.g., if the Hamiltonian family is always Hermitian w.r.t.\ a given inner product on $V$, then $\En(H)$ is trivial as $\Spec(H)$ is,\footnote{An alternative argument: if the energies of $H$ are real, then $\En(H)$ is a pull-back of $\R^{\underline{n}}\to\binom{\R}{n}$, hence trivial by Proposition~\ref{prop:orderable -> trivial}.} and the structure group can be reduced to the trivial group 1. Intermediate cases are also possible, and express in what way the energy bands are connected. For example, if $k$ energy bands are disconnected from the remaining $n-k$, clearly the structure group $S_n$ can be reduced to $S_k\times S_{n-k}$. Practically, this can be achieved by restricting to tuples $\tilde{E}$ in which these $k$ energy bands appear only in the first $k$ entries.

\begin{Example} %\label{exmp:EP holonomy}
 Let us continue the study of the standard EP2 system as we described in Example~\ref{ex:start EP2}. We found $N(H)=\C\setminus \{\pm {\rm i}\}$, which means $\pi_1(N(H),0)$ has two generators. Define generators $a_\pm$ to wind once around $\pm {\rm i}$ in positive direction. We found in Example~\ref{exmp:monodromy action EP2} that both generators exchange the eigenvalues. Let us label the spectrum $\Spec(H(0))=\{\pm1\}$ as $(+1,-1)$. Then the permutations arising from the EPs are expressed by the homomorphism
 \begin{align*}
 \pi_1(N(H),0) &\to \Hol^{\En(H)}(0,(+1,-1)),
 \\
 a_\pm &\mapsto (12).
 \end{align*}
 Clearly, this confirms that a permutation occurs if and only if a loop decomposes into an odd number of $a_\pm$. As seen in Example~\ref{ex:EP2 yields (12)}, if an extra level is present, the exact form of the permutation does indeed depend on the labelling. In addition, we observe that for the extended system $H'$ presented there, the principal $S_3$-bundle $\En(H')$ reduces to a $S_2$-bundle, which is isomorphic to $\En(H)$.
\end{Example}

We note that the space $\En(V)$ has appeared in previous papers. We already remarked that~\cite{Tanaka2017PathEvolution} mentioned the relevance of the monodromy action in the theory of EPs. The covering used there is of the form $\En(H)$, using the one-to-one correspondence between eigenvalues and eigenprojectors for non-degenerate operators. We also recognize $\En(V)$ as the space $\widetilde{\mathfrak{M}}$ defined in \cite{Mehri-Dehnavi2008GeometricInterpretation}. However, where \cite{Mehri-Dehnavi2008GeometricInterpretation} views $\En(V)$ primarily as an extended parameter space, we consider it rather as the bundle over $N(V)$ which yields the permutations due to EPs by its holonomy.

\subsubsection{Frame bundle of the eigenvector bundle}

We now proceed by considering the frame bundle of $\pi_{\lambda v}\from \Eig(V) \to N(V)$. This is a semi-principal $\C^\times$-bundle by Theorem~\ref{thm:Eig(V) bundle over N(V)} and endowed with the $\C^\times$-connection $\omega$ as in Proposition~\ref{prop:omega Ctimes connection}. Applying the frame bundle functor then yields the following.

\begin{Proposition}
 The frame bundle of $\Eig(V)$ is the principal $\struc$-bundle
 \begin{equation*} %\label{eq:eigenframe bundle}
 \EigFr(V)=\set{\big(A,\tilde{f}\big)\in N(V)\times \Fr(V)}{\tilde{f}\text{ is an eigenframe of }A},
 \end{equation*}
 where the element $\big(A,\tilde{f}\big)$ can also be written as $\big(A,\tilde{\lambda},\tilde{f}\big)$ with $\tilde{\lambda}$ listing the eigenvalues of $A$ in the order of $\tilde{f}$. The $\struc$-action on $\EigFr(V)$ reads
 \begin{equation*} %\label{eq:struc action}
 (\tilde{z},\sigma) \cdot\big(A,\tilde{\lambda},\tilde{f}\big) = \big(A,\sigma\tilde{\lambda},\tilde{z}\big(\sigma\tilde{f}\big)\big)
 \end{equation*}
 and the projection is
 \begin{align*}
 \pi_{\lambda v}!\from\ \EigFr(V) &\to N(V),
 \\
 \big(A,\tilde{\lambda},\tilde{f}\big) &\mapsto A.
 \end{align*}
 Furthermore, $\EigFr(V)$ is naturally endowed with the $\struc$-connection
 \begin{equation*}
 \omega!_{\left(A,\tilde{\lambda},\tilde{f}\right)}=\omega_{(A,\lambda_1,f_1)} \oplus \dots \oplus \omega_{(A,\lambda_n,f_n)}.
 \end{equation*}
\end{Proposition}

Recall that we already found the principal $\struc$-bundle $\Xi\from \Fr(V)\times \Cun \to N(V)$ in Section~\ref{sec:param of non-deg}. We observe that $\Xi$ and $\pi_{\lambda v}!$ are reformulations of each other, i.e., they are canonically isomorphic bundles.
\begin{Lemma}
 The map
 \begin{align*}
 \EigFr(V) &\to \Fr(V)\times \Cun,
 \\
 \big(A,\tilde{\lambda},\tilde{f}\big) &\mapsto \big(\tilde{f},\tilde{\lambda}\big)
 \end{align*}
 is a canonical isomorphism of bundles over $N(V)$. Hence, in particular, $\EigFr(V)$ is canonically isomorphic to $\Fr(V)\times \Cun$ as a $\struc$-manifold, and consequently $\En(V)$ is canonically dif\-feo\-morphic to $\PFr(V)\times \Cun$.
\end{Lemma}

Again, given a Hamiltonian family $H$, it is more natural to study the pull-back bundle. This yields the space $\EigFr(H)$, whose elements we write as $\big(x,\tilde{E},\tilde{\psi}\big)$. The connection form $\omega!$ also carries over to the connection form $\omega_H!$.
Note that $\EigFr(H)$ may admit a reduction of the structure group, e.g., to $U(1)^n$ when $H$ is Hermitian.
We can summarize the pull-back and relation with $\Xi$ in diagram~\eqref{diag:EigFr(V) and related}.
\begin{equation} \label{diag:EigFr(V) and related}
 \begin{tikzcd}[column sep=1.5em]
 \EigFr(H) \ar{r}\ar{d}{\pi_{\lambda v}^H!} & \EigFr(V) \ar{r}{\sim}\ar{d}{\pi_{\lambda v}!} & \Fr(V)\times \Cun \ar{d}{\Xi} & \big(x,\tilde{E},\tilde{\psi}\big) \ar[mapsto]{r}\ar[mapsto]{d} & \big(H(x),\tilde{E},\tilde{\psi}\big) \ar[mapsto]{r}\ar[mapsto]{d} & \big(\tilde{\psi},\tilde{E}\big)\ar[mapsto]{d}
 \\
 N(H) \ar{r}{H} & N(V) \ar{r}{\id} & N(V), & x \ar[mapsto]{r} & H(x) \ar[mapsto]{r} & H(x).
 \end{tikzcd}
\end{equation}

We can now view the holonomy interpretation of the map on eigenstates as follows. Again, we fix a family $H$ and a loop $\gamma$ in $N(H)$ based at some reference $x_0$. We then pick an eigenframe~$\tilde{\psi}$ of $H(x_0)$, which fixes the tuple $\tilde{E}$ of corresponding energies. That is, we pick a point $\big(x_0,\tilde{E},\tilde{\psi}\big)$ in the fiber of $\EigFr(V)$ above $x_0$. This point fixes a unique lift of $\gamma$ to $\EigFr(V)$, which traces all states in the tuple $\psi$ simultaneously. The end point of this lift then equals $(\tilde{z},\sigma)\big(x_0,\tilde{E},\tilde{\psi}\big)$, where $(\tilde{z},\sigma)\in \struc$ we obtained as before. Hence it lies in the holonomy group
\begin{equation*}
 \Hol^{\EigFr(H)}\big(x_0,\tilde{E},\tilde{\psi}\big)=\set{(\tilde{z},\sigma)\in \struc}{\big(x_0,\tilde{E},\tilde{\psi}\big)\sim (\tilde{z},\sigma)\big(x_0,\tilde{E},\tilde{\psi}\big)},
\end{equation*}
where the equivalence relation $\sim$ means that the points can be connected by a horizontal path in $\EigFr(H)$, or equivalently, that there is a path in $N(H)$ whose lift to $\EigFr(H)$ through one point ends at the other. Again, we consider the holonomy group at a point in the total space rather than in the base space.

Clearly, $\sigma$ is the same permutation as we would find when looking only at the energies. This we already found in Proposition~\ref{prop:Pgamma reduces to pgamma}, which is based on the map $\pi_v \from \Eig(V) \to \Spec(V)$. However, we can also phrase this on the frame bundles as follows. Namely, the frame bundle functor turns~$\pi_v$ into the bundle map $\pi_v!\from \EigFr(V) \to \En(V)$, $\big(A,\tilde{\lambda},\tilde{f}\big)\mapsto \big(A,\tilde{\lambda}\big)$. This map is equivariant w.r.t.\ the canonical quotient $\struc \to S_n$, and is compatible with the connections. Fixing a point $\big(A,\tilde{\lambda},\tilde{f}\big)$, this quotient reappears as a map on holonomy groups $\Hol^{\EigFr(V)}\big(A,\tilde{\lambda},\tilde{f}\big) \to \Hol^{\En(V)}\big(A,\tilde{\lambda}\big)$. This map simply extracts $\sigma$, which confirms that $\sigma$ is the holonomy from $\En(V)$. We summarize the above findings as follows.

\begin{Lemma}
 Given a loop $\gamma$ in $N(H)$ based at $x_0$, the element $(\tilde{z},\sigma)$ expressing $P_\gamma$ w.r.t.\ the eigenframe $\tilde{\psi}$ of $H(x_0)$ is the holonomy element at the point $(x_0,\tilde{\psi})$ in the frame bundle $\EigFr(H)$. Writing $\tilde{E}$ for the corresponding ordering of $\Spec(H(x_0))$, $\sigma$ is the permutation representing $p_\gamma$ w.r.t.\ $\tilde{E}$.
\end{Lemma}

The element $(\tilde{z},\sigma)$ thus lies in the group $\Hol^{\EigFr(H)}(x_0,\tilde{\psi})$. We can view the picking of $\tilde{\psi}$ as a gauge freedom, just like the picking of $\tilde{E}$ earlier. A change of gauge thus corresponds to switching holonomy group, with the usual similarity transformation of the holonomy matrix. The gauge invariant in this case is the map $P_\gamma$. Indeed, $P_\gamma$ is defined without using any gauge, hence is trivially independent of gauge.

Let us then take a closer look at the phase factors $\tilde{z}$ appearing in the holonomy matrix. Given the parallel transport formalism, each of these may be called a ``geometric'' phase factor. However, we also see that they are gauge dependent. Namely, the order gets permuted by changing the order of the eigenframe $\tilde{\psi}$, and scaling of the individual states in $\tilde{\psi}$ affects the phase factors in $\tilde{z}$ corresponding to non-cyclic evolution. On the other hand, $\tilde{z}$ is subject to constraints coming from the characteristic phases of $P_\gamma$. Parametrizing $z_j=\exp({\rm i}\alpha_j)$, if we follow a state until it returns to its original ray, its acquired phase is $\sum_{\text{cycle}}\alpha_i$, where we sum over the indices belonging to the corresponding cycle. Hence $\sum_{\text{cycle}}\alpha_i$ must be equal to the characteristic phase of the cycle, independent of the chosen gauge. Note however that this does not yield a canonical phase for a non-cyclic state after following $\gamma$ once; in one gauge all $\alpha_i$ in the cycle could be equal, whereas in another all except one could vanish.

In terms of the holonomy matrix, we can express this as follows. By permuting $\tilde{\psi}$ the phase factors can be moved within the matrix, which is subject to the underlying cycle structure of the permutation. The matrix is block-diagonal if and only if states from the same minimal union are listed adjacently. The size of a block is $k\times k$, with $k$ the length of the corresponding cycle. The $k$ individual phase factors within this block can be gauged to arbitrary values, with the only condition that their product is fixed.

Let us treat in more detail how a seemingly Abelian connection $\omega_H!$ can still induce non-Abelian holonomy. That is, given the technique of path-ordered exponential and the diagonal form of $\omega_H!$, one may wonder how the end result can be non-diagonal.
In short, the reason is that if the loop $\gamma$ lifts to non-cyclic states, then any instantaneous eigenframe returns with a~twist, and this twist yields an additional permutation matrix making the theory non-Abelian. That is, in order to write $\omega_H!$ as a diagonal gauge potential, in practice one would need to extend the initial eigenframe $\tilde{\psi}$ to a path of instantaneous eigenframes over $\gamma$. Upon returning to $x_0$, one cannot use the initial eigenframe $\tilde{\psi}$ again due to the presence of non-cyclic states. Consequently, one needs to correct for the difference between final and initial eigenframe. This basis transformation can be expressed by an element in $\struc$, and thus permutation matrices enter the end result. Naturally, this issue can only occur when $\En(H)$ is non-trivial and $\gamma$ is a~non-contractible loop, and so is completely absent in case $\pi_1(N(H),x_0)$ is trivial.

Let us also compare $\EigFr(V)$ to earlier work in \cite{Mehri-Dehnavi2008GeometricInterpretation}, which is likewise aimed to provide a~holonomy description for geometric phases in the presence of EPs.
There, one considers cyclic states under the adiabatic evolution due to a change of a Hamiltonian, which may be non-Hermitian. The general geometric model consists of $n$ complex line bundles $L^1,\dots,L^n$ over~$\En(V)$, there called $\widetilde{\mathfrak{M}}$, where each line bundle corresponds to a certain self-chosen part of the energy bands. A main disadvantage of this approach is that the connection information is spread over the $L^k$. Hence, calculation of a lift can in general not be done in a single $L^k$, i.e., multiple line bundles need to be used. That is, formally the $L^k$ do not support a holonomy interpretation for geometric phases in the presence of EPs. In contrast, we found that $\EigFr(V)$ is a single bundle allowing to lift a path directly from $N(V)$, which can moreover be used to treat non-cyclic states.

It is also the case that the theory of $\EigFr(V)$ provides a geometric model for the off-diagonal phases reported by Manini and Pistolesi \cite{Manini2000Off-DiagonalPhases}, and extends the generalized permutation matrices found by Tanaka, Cheon and Kim \cite{Tanaka2012GaugeCircuits}. The generalization of their work to non-Hermitian Hamiltonians can be obtained from $\EigFr(H)$ in the following way.

First, let us clarify that there are 3 main classes of paths that are considered in this theory. In~the above, we considered loops $\gamma$ only, i.e., the final point equals the initial point~$x_0$. In~contrast, \cite{Manini2000Off-DiagonalPhases} starts by considering general paths, so that the final point $x_1$ need not equal~$x_0$. They then consider a more special case, which is also treated in \cite{Tanaka2012GaugeCircuits} and there called an ``adiabatic loop''. This term refers to a path $\gamma$ for which $x_1$ need not equal~$x_0$, but still $\Eig(H(x_0))=\Eig(H(x_1))$ (as subsets of $V$, strictly speaking not as fibers of $\Eig(H)$). We~thus observe that an ``adiabatic loop'' need not be a ``loop'' in the sense we use, and forms an intermediate subset:
\begin{equation*}
 \left\{\begin{gathered}\text{loops;}\\ x_0=x_1\end{gathered}\right\} \subset \left\{\begin{gathered}\text{``adiabatic loops'';}\\ \Eig(H(x_0))=\Eig(H(x_1))\end{gathered}\right\} \subset \left\{\begin{gathered}\text{general paths;}\\ \text{no relation } x_0, x_1\end{gathered}\right\}\!.
\end{equation*}

Let us now demonstrate that $\EigFr(V)$ also describes the theory of the general paths, so that we obtain the ``adiabatic loops'' as a special case. Let $\gamma$ be any path in $N(H)$, so belonging to the most right set in the above comparison. Lifting to $\EigFr(V)$, it induces the parallel transport map
\begin{equation*}
 \tilde{P}_\gamma \from\ \EigFr(H(x_0)) \to \EigFr(H(x_1)).
\end{equation*}
More explicitly, given an initial eigenframe $\tilde{\psi}$ of $H(x_0)$, one obtains a lift $\tilde{\Gamma}$ of $\gamma$ starting at $\tilde{\psi}$, and the final frame is $\tilde{\psi}'=\tilde{P}_\gamma(\tilde{\psi})$. However, as $\tilde{\psi}'$ and $\tilde{\psi}$ are (generally) not in the same fiber, the two frames cannot be compared using the bundle structure of $\EigFr(H)$. The solution is to consider the frames by themselves; one neglects the fact that they are associated to different values of the system parameters. Clearly, any two frames of $V$ are related by a unique matrix; in this case there is a unique $U\in \GL(n,\C)$ such that
\begin{equation*}
 \tilde{\psi}'=U\tilde{\psi}.
\end{equation*}
Mathematically, this argument is captured by employing the canonical projection $\EigFr(V) \to \Fr(V)$, and using that $\Fr(V)$ is a $\GL(n,\C)$-torsor. We also remark that if both $\tilde{\psi}$ and $\tilde{\psi}'$ are orthonormal w.r.t.\ some inner product, then $U$ is indeed unitary, and can be calculated by taking inner products between the states, as used in \cite{Manini2000Off-DiagonalPhases}.

\looseness=1 Of course, the matrix $U$ is not unique; if we change $\tilde{\psi}$ to $G\tilde{\psi}$, with $G\in \struc$ a generalized permutation matrix, then $\tilde{\psi}'=\tilde{P}_\gamma\big(\tilde{\psi}\big)$ becomes $\tilde{P}_\gamma\big(G\tilde{\psi}\big)=G\tilde{P}_\gamma\big(\tilde{\psi}\big)=G\tilde{\psi}'$, and so $U$ becomes $GUG^{-1}$. That is, $U$ is unique up to conjugation by $\struc$. When looking for invariants, if we write $G=(\tilde{g},\tau)$, we may set $\tau=\id$; non-trivial $\tau$ simply rearrange the elements of~$U$. It~follows that $U_{ij}$ becomes $g_iU_{ij}g_j^{-1}$. Hence, one obtains an invariant quantity for each cycle in~$S_n$ by taking the corresponding product, e.g., $(12)$ yields the invariant $U_{12}U_{21}$ and $(123)$ yields the invariant $U_{12}U_{23}U_{31}$, in addition to the diagonal elements $U_{kk}$ which are also invariant. Whenever such a product is non-zero, its argument is a gauge-invariant phase, as reported in~\cite{Manini2000Off-DiagonalPhases}.

The ``adiabatic loops'' are then the special case where $U$ is also a generalized permutation mat\-rix, which was studied further in \cite{Tanaka2012GaugeCircuits}. The similarity of this case with our study of the holo\-nomy on $\EigFr(V)$ is straightforward to explain; in both cases $\EigFr(H(x_0))=\EigFr(H(x_1))$, so that we are studying an automorphism of a semi-torsor. The method above, based on the matrix $U$, then reduces to the argument of Section~\ref{sec:explicit perms and hol mats}, and invariant phases follow from Proposition~\ref{prop:semi-torsor invariant subspaces}.\looseness=1

\section{Discussion} \label{sec:discussion}

In this paper we introduced a framework that facilitates the mathematical description of the adiabatic quantum mechanics of finite-dimensional, non-degenerate but otherwise arbitrary and hence not necessarily Hermitian Hamiltonians.
We started with a preparatory study of the space of non-degenerate operators, which was summarized in diagram~\eqref{eq:summarizing diagram}. The following step was to find the geometry behind the adiabatic change of energies of a Hamiltonian family~$H$. This is captured by the covering space formed by the energy bands, where we restrict to system parameter values where~$H$ is non-degenerate. Moreover, the specific covering needed for $H$ we found as the pull-back of an abstract yet explicit model covering, which makes the covering property and hence bundle structure immediate for any smooth Hamiltonian fa\-mily~$H$.\looseness=1

\looseness=1 We continued by introducing a bundle for the eigenstates. A remarkable property of this bundle is that it is not a principal bundle, but has the structure of a semi-principal bundle, which is described in \cite{Pap2020FramesTheory}. That is, each of its fibers consists of a collection of eigenrays rather than a single eigenray. We found a natural connection on this bundle, which provides a parallel transport description of the adiabatic geometric phase as found by Garrison and Wright \cite{Garrison1988ComplexSystems}. Because of the semi-principal structure, this incorporates non-cyclic states as well. This formalism hence extends previous holonomy formulations, which are based on principal bundles. Among these are the first holonomy interpretation reported by Simon \cite{Simon1983HolonomyPhase}, which treats the Hermitian case with separate energy bands, and the formalism in~\cite{Mehri-Dehnavi2008GeometricInterpretation}, which considered the non-Hermitian case restricted to cyclic states. We furthermore found that the Aharonov--Anandan formalism~\cite{Aharonov1987PhaseEvolution} for general Hermitian systems can be obtained by a reduction from the present formalism using a~given inner product. Moreover, we found that the dynamical phase can be included as a non-geometric addition, which means that the full adiabatic evolution of a state can be calculated naturally on this bundle. We then obtained a generalized quantum geometric tensor, which we found to be directly related to covariant derivatives.

We highlight that the presented formalism, as it allows one to also describe non-cyclic states by means of parallel transport, is a natural formalism for the study of exceptional points. The description including geometric phases is then given on the full eigenstate bundle, whereas the permutations of energies can already be accurately described with the monodromy theory of the energy band covering. Both of these spaces have a semi-principal bundle structure, which is the key to incorporate non-cyclic behavior. However, by switching to a multi-state approach, one can describe the permutations of energies and the accompanying geometric phase factors simultaneously in a more explicit way, namely via holonomy matrices. This is given by a~straightforward application of the frame bundle technique as described in~\cite{Pap2020FramesTheory}. This provides a more explicit holonomy description of state evolution, also in the non-cyclic case enabled by exceptional points.

\appendix

\section{Bundle structure of tuples of distinct complex numbers}
\label{sec:distinct numbers}

If one considers $n$ distinct complex numbers, one can do so with or without an ordering of these numbers. Let us start with the ordered one. This brings us to the space
\begin{equation*}
 \Cun=\set{(z_1,\dots,z_n)\in \C^n}{i\ne j \implies z_i\ne z_j},
\end{equation*}
where we borrow the notation of falling factorials; if we would replace $\C$ by a finite set $F$, then $|F^{\underline{n}}|=|F|(|F|-1)\cdots(|F|-n)$, which in combinatorics is known as the falling factorial $|F|^{\underline{n}}$. The space $\Cun$ is a submanifold of $\C^n$, and more is true.
\begin{Lemma}
 The space $\Cun$ is an $($algebraic$)$ open and dense submanifold of $\C^n$.
\end{Lemma}

\begin{proof}
 For $i\ne j$, define the function $m_{ij}\from \C^n\to \C$, $(z_1,\dots,z_n)\mapsto z_i-z_j$, and write $D(m_{ij})=\set{(z_1,\dots,z_n)\in \C^n}{m_{ij}(z_1,\dots,z_n)\ne0}$. Then $\Cun=\cap_{i\ne j} D(m_{ij})$ expresses the space as a~finite intersection of (algebraic) open subsets, hence is open in $\C^n$ and so a submanifold. It~is also dense; given a tuple with repeated numbers, these become unequal by an arbitrarily small perturbation.
\end{proof}

We can now proceed to sets of $n$ distinct complex numbers, which is the unordered variant of the ordered tuple case as just discussed. The idea is to view the unordered sets as a quotient of the tuple space. This quotient is specified by a permutation action; upon identifying rearranged tuples, clearly the ordering vanishes. The action is the usual permutation of slots, i.e., for $\sigma\in S_n$ and $(z_1,\dots,z_n)\in \Cun$ we have
\begin{equation*}
 \sigma \cdot (z_1,\dots,z_n) = (z_{\sigma^{-1}(1)},\dots,z_{\sigma^{-1}(n)}).
\end{equation*}
This action is smooth as it is the restriction of the smooth permutation action on $\C^n$ to an invariant submanifold. Taking the quotient results in the space of subsets of $\C$ with $n$ distinct complex numbers, which we write as
\begin{equation*}
 \binom{\C}{n}:=\Cun/S_n=\set{S\subset \C}{|S|=n}.
\end{equation*}
Again, we borrow our notation from combinatorics. We write $q$ for the quotient map $\Cun \to \binom{\C}{n}$. Clearly, the fiber of $q$ above a certain set in $\binom{\C}{n}$ consists of all the possible orderings of its elements. Via $q$ one obtains a unique manifold structure on $\binom{\C}{n}$, and in this way $q$ is principal bundle.
\begin{Proposition}
 There is a unique manifold structure on $\binom{\C}{n}$ such that the quotient map $q\from \Cun \to \binom{\C}{n}$ is a surjective submersion. The permutation action on $\Cun$ thus induces the principal bundle
 \begin{equation*}
 \begin{tikzcd}
 S_n \ar{r} & \Cun \ar{r}{q} & \binom{\C}{n}.
 \end{tikzcd}
 \end{equation*}
 Moreover, this bundle is non-trivial for $n>1$.
\end{Proposition}

\begin{proof}
 Clearly the $S_n$-action on $\Cun$ is free and proper, hence the quotient map defines a principal bundle. Assume it is trivial for $n>1$ for contradiction. This would imply that $\Cun$ and~$\binom{\C}{n}\times S_n$ are homeomorphic. However, as $\Cun$ and $\binom{\C}{n}$ are path-connected but $S_n$ is not, this cannot be.
\end{proof}

The non-triviality of the bundle is related to the topology of $\C$ not being compatible with a~total ordering. To expand on this, observe that for any manifold $M$ one can define the space of distinct tuples $M^{\underline{n}}$. This similarly has a $S_n$-action defined on it, and one obtains a principal bundle over the space $\binom{M}{n}$ of unordered subsets of $n$ elements. If $M$ has an order topology given by a total ordering, such as $\R$, this bundle is trivial. Indeed, given any element $S \in \binom{M}{n}$, then there is a standard way to list the elements of $S$, e.g., from lowest to highest. This induces triviality as follows.
\begin{Proposition} \label{prop:orderable -> trivial}
 If a manifold $M$ is orderable, then the principal $S_n$-bundle $M^{\underline{n}} \to \binom{M}{n}$ is trivial for all $n$.
\end{Proposition}
\begin{proof}
 Let $\leq$ be a total ordering on $M$ compatible with the topology. Define an ordering map $o \from \binom{M}{n} \to M^{\underline{n}}$ by the condition that
 \begin{equation*}
 o(\{m_1,\dots,m_n\})=(m_1,\dots,m_n),\qquad \text{where}\quad m_1<\dots <m_n.
 \end{equation*}
 As all the $m_i$ are distinct this is well-defined, and by assumption on the topology of~$M$ also continuous. It follows that $o$ is a global section of the principal bundle, which is thus trivial.
\end{proof}

\subsection*{Acknowledgements}
The authors thank the anonymous referees whose careful remarks contributed to the quality of the paper.

\pdfbookmark[1]{References}{ref}
\LastPageEnding

\end{document}